\documentclass[a4paper,twocolumn,11pt,unpublished]{quantumarticle}
\pdfoutput=1
\usepackage[utf8]{inputenc}
\usepackage[english]{babel}
\usepackage[T1]{fontenc}
\usepackage{amsmath}
\usepackage{hyperref}

\usepackage{tikz}
\usepackage{lipsum}

\usepackage{braket}
\usepackage{quantikz}
\usepackage{amssymb}
\usepackage{amsthm}

\newtheorem{theorem}{Theorem}

\newtheorem{proposition}[theorem]{Proposition}
\newtheorem{lemma}[theorem]{Lemma}
\newtheorem{definition}[theorem]{Definition}

\numberwithin{equation}{section}

\begin{document}

\title{Categorical Tensor-Graph Semantics for Quantum Algorithms }

\author{Naihong Hu}
\email{nhhu@math.ecnu.edu.cn}
\affiliation{School of Mathematical Sciences, MOE Key Laboratory of Mathematics and Engineering Applications $\&$ Shanghai Key Laboratory of PMMP, East China Normal University, Shanghai 200241, China}
\orcid{0000-0001-8162-8155}
\author{Ruining Li}
\email{737857625@qq.com}
\affiliation{School of Mathematical Sciences, MOE Key Laboratory of Mathematics and Engineering Applications $\&$ Shanghai Key Laboratory of PMMP, East China Normal University, Shanghai 200241, China}
\author{Futao Wang}
\email{52280155005@stu.ecnu.edu.cn}
\affiliation{School of Mathematical Sciences, MOE Key Laboratory of Mathematics and Engineering Applications $\&$ Shanghai Key Laboratory of PMMP, East China Normal University, Shanghai 200241, China}

\maketitle

\begin{abstract}
  
This paper investigates foundational quantum computing protocols from the intuitive perspective of categorical tensor-graph semantics within the category \textbf{FHilb}. While conventional Hilbert-space formalisms often conceal the structural nature of quantum algorithms behind high-dimensional matrix operations, the topological framework directly encodes algorithmic functionalities into their graphical skeletons. We provide a comprehensive topological reinterpretation of the Bernstein--Vazirani and Simon algorithms, demonstrating how topological transformations distill their core mathematical essence and clarify the operational mechanisms of oracles. Going beyond the standard qubit model, we construct explicit representations for the qutrit-adapted topological Deutsch--Jozsa and single-shot Grover algorithms. In particular, we establish a necessary and sufficient condition for the single-shot Grover search. We further implement CNOT gates via complementary Frobenius structures and investigate a diagrammatic decomposition scheme for the W-state preparation protocol. By bridging tensor category theory with practical quantum algorithmic design, this work furnishes a composable, scalable diagrammatic toolkit essential for automated circuit optimization across the evolving quantum hardware ecosystem.
\end{abstract}

\section{Introduction}

Since its inception by Feynman and Deutsch in the 1980s, quantum computing has evolved from abstract models toward physical realization. Presently, platforms by Google, IBM, IonQ, and Rigetti have achieved tens to hundreds of qubits—Google claimed "quantum supremacy" in 2019, while IBM launched the 1000+ qubit Condor processor, marking the NISQ era  \cite{NC11, BK23}. However, traditional Hilbert-space formal methods, though rigorous, involve high-dimensional matrix operations that obscure the intuitive structure of entanglement, superposition, and interference, and struggle with combinatorial complexity. As hardware advances rapidly, more intuitive and compositional tools are urgently needed.

Categorical quantum mechanics addresses this challenge. Under severe NISQ constraints—limited qubits, short coherence times, and high noise—minimizing circuit depth and gate count is crucial. ZX-calculus transforms complex circuit equivalences into intuitive topological operations via graphical rewriting; phase teleportation enables non-local cancellation of non-Clifford phases, significantly reducing costly T-gates \cite{KvdW20}. Recent work combining reinforcement learning with ZX-calculus achieves continuous optimization at 80-qubit, 2100-gate scales, providing a new automated paradigm for large-scale compilation. Open-source libraries PyZX and Quantomatic have enabled practical deployment and commercial integration  \cite{KZ15}. Graphical languages also naturally describe lattice surgery and error mitigation protocols for fault-tolerant computing  \cite{DH19, DPZ24}.

This methodology is expanding across domains. In quantum natural language processing, the DisCoCat model unifies syntax and semantics via tensor networks, leveraging entanglement for semantic composition  \cite{CCS10, MGT20}. In quantum machine learning, diagrammatic differentiation enables automatic differentiation for parameterized circuits \cite{TYF21}. During 2024–2025, advances combining graph neural networks with reinforcement learning (e.g., PPO-based policy optimization) and dynamic grouping strategies have further enhanced scalability.

In summary, categorical quantum mechanics and diagrammatic semantics provide: (1) new mathematical perspectives linking quantum computing to topology, logic, and computer science \cite{HV19,  KM19,CDKW12, Vic13}; (2) intuitive, compositional, and automatable tools for algorithm design; and (3) full-process support for NISQ-era compilation while opening cross-disciplinary pathways with AI and NLP. As hardware scales expand, systematic development of these theories and tools is critical for transitioning quantum computing from laboratory to practice.

The quest for an intuitive and rigorous quantum formalism is deeply rooted in the mathematical development of tensor categories \cite{Tur94, BK01, EGNO05}. In pure mathematics, the theory of tensor categories has not only achieved significant success in areas such as topological invariants but has also fundamentally proven the strict algebraic equivalence between diagrammatic reasoning and planar topological deformations. This profound mathematical foundation validates categorical graphical semantics, ensuring that topological deformations in quantum computing are not merely heuristic illustrations, but exact algebraic evaluations. 

Building upon Abramsky and Coecke's pioneering Categorical Quantum Mechanics (CQM) \cite{AC04} and the subsequent ZX-calculus \cite{CK17}, diagrammatic semantics have established a powerful foundation for quantum protocols. Furthermore, the theoretical boundaries of this framework continue to expand, as demonstrated by Majid's formulation of a quantum and braided ZX-calculus that generalizes these diagrammatic techniques to braided monoidal categories \cite{Maj23}. Directly inspired by Vicary's foundational exploration of the topological structure of quantum algorithms \cite{Vic13}, our work leverages tensor-graph semantics within the category \textbf{FHilb} to fundamentally dissect and reinterpret complex quantum computation. Rather than treating graphical rules merely as simplification tools, we utilize rigorous topological deformations to distill the algorithmic core of the Bernstein--Vazirani and Simon algorithms, the qutrit-adapted generalized Deutsch--Jozsa and single-shot Grover algorithms, as well as complex multi-particle entanglement protocols like the W-state preparation.

This study employs categorical topological language to describe quantum algorithms and quantum mechanics, revealing the internal mechanism of quantum advantage through tensor diagrams and demonstrating operational procedures via graphical expressions.

In Section 2, we focus on the tensor diagram semantics of the Bernstein-Vazirani algorithm. We introduce the algorithm's formal definition and quantum circuit, provide its tensor diagram semantics, and prove their correctness using transformation principles. Topological transformations are then applied to derive the simplest graphical form, and a comparison with the original circuit reveals the algorithm's mathematical essence from the tensor diagram perspective.

In Section 3, we address the Simon algorithm analogously. Tensor diagram semantics are given, verified, and simplified via topological transformations, with a comparison to the original circuit to analyze its mathematical connotation. A comparative analysis of both algorithms' simplified semantics explains their similarities and essential differences.

In Section 4, going beyond the binary qubit paradigm, we investigate multi-valued quantum systems. We provide the qutrit-adapted generalized Deutsch-Jozsa algorithm and the generalized single-shot Grover algorithm, mapping their execution directly to topological diagrams.

In Section 5, we analyze the controlled-NOT gate and quantum entanglement from the perspective of categorical composition. We provide an alternative diagrammatic proof for the properties of complementary Frobenius structures. Furthermore, we translate the standard W-state preparation circuit into the ZX-calculus framework and perform a systematic simplification using phase gadgets, local complementation and pivot rule. We demonstrate that this topological reduction yields an expression distinct from existing standard formulations.

\section{Topological Bernstein-Vazirani}\label{pre}
\subsection{Introduction}
In this section, we will use the knowledge of categorical quantum mechanics mentioned in Appendix A to reinterpret the Bernstein-Vazirani algorithm from a tensor diagram perspective. This formulation avoids the complex computational processes in traditional circuit diagrams, converts the connections between different qubits into topological structures in tensor diagrams, and directly attributes the algorithm's functionality to its topological skeleton.

The Bernstein-Vazirani algorithm considers the analysis of a function with a finite number of input bits and a single output bit. Suppose we have a hidden binary string $s\in \{0, 1\}^{n}$, and the function $f: \{0, 1\}^{n} \rightarrow \{0, 1\}$ satisfies $f(x)=x \cdot s$. The Bernstein-Vazirani problem is to find the string $s$.

In the classical case, to obtain the value of an $n$ -bit string $s$, we need to measure $n$ times. However, the quantum Bernstein-Vazirani algorithm only needs to access the function $f$ once to obtain the values of all $n$  bits. Its quantum circuit is as follows and inspired by  \cite{Vic13}, we give the topological version of Bernstein-Vazirani algorithm.
\begin{center}
\begin{equation}\label{BV2}
\begin{tikzpicture}[thick, scale=0.6]

    \draw (0,2) -- (2.3,2);     
    \draw (3.5,2) -- (4.2,2);   
    \draw (5.8,2) -- (6.3,2);   
    \draw (7.5,2) -- (10.5,2);  

    \draw (0,0) -- (1.5,0);     
    \draw (2.3,0) -- (2.8,0);   
    \draw (3.6,0) -- (4.2,0);   
    \draw (5.8,0) -- (7.8,0);   
    \draw (8.6,0) -- (10.5,0);  

    \node at (-0.5,2) {$|0\rangle$};
    \node at (-0.5,0) {$|0\rangle$};

    \draw (0.3,1.8) -- (0.7,2.2);
    \node at (0.6,2.4) {\small $n$};

    \draw (2.3,1.6) rectangle (3.5,2.4);
    \node at (2.9,2) { $H^{\otimes n}$};

    \draw (1.5,-0.4) rectangle (2.3,0.4);
    \node at (1.9,0) {$H$};
    \draw (2.8,-0.4) rectangle (3.6,0.4);
    \node at (3.2,0) {$Z$};

    \draw (4.2,-0.5) rectangle (5.8,2.5);
    \node at (5,1) {\Large $U_f$};

    \draw (6.3,1.6) rectangle (7.5,2.4);
    \node at (6.9,2) {$H^{\otimes n}$};
    \draw (7.8,-0.4) rectangle (8.6,0.4);
    \node at (8.2,0) {$H$};

    \draw (10.5,1.6) rectangle (11.3,2.4);
    \draw (10.55,1.8) arc (150:30:0.4); 
    \draw[->] (10.7,1.7) -- (11.1,2.2); 

    \node at (11,0) {$|1\rangle$};

    \begin{scope}[every node/.style={inner sep=1pt}]
        \node at (1.1,-1) (p0) {$|\psi_0\rangle$};
        \draw[->] (p0) -- (1.1,-0.5);

        \node at (3.8,-1) (p1) {$|\psi_1\rangle$};
        \draw[->] (p1) -- (3.8,-0.5);

        \node at (6,-1) (p2) {$|\psi_2\rangle$};
        \draw[->] (p2) -- (6,-0.5);

        \node at (7.5,-1) (p3) {$|\psi_3\rangle$};
        \draw[->] (p3) -- (7.5,-0.5);

        \node at (9.5,-1) (p4) {$|\psi_4\rangle$};
        \draw[->] (p4) -- (9.5,-0.5);
    \end{scope}
\end{tikzpicture}
\end{equation}
\end{center}
\begin{definition}[Topological Bernstein-Vazirani] \label{BVdef}
Let $A= ({\mathbb{C}}^2)^{\otimes n} , B= {\mathbb{C}}^2$. For an $n$-bit string $s$, take the computational basis on $(\mathbb{C}^2)^{\otimes n}$ as the Frobenius structure $(A, \mu_{B} ,  \eta_{B})$ on $A$, and the computational basis and $\sqrt{2}\mathrm{X}$ basis on $\mathbb{C}^2$ as the Frobenius structures $(B, \mu_{B} ,  \eta_{B})$ and $(B, \mu_{W} ,  \eta_{W})$ on $B$. Let $I \xrightarrow{b} B$ be the state $\binom{1 / \sqrt{2}}{-1 / \sqrt{2}}$. Thus for a given function $A \xrightarrow{f} B$, the tensor diagram of the Bernstein-Vazirani algorithm in $\mathbf{FHilb}$ is as follows:
\begin{center}
\begin{equation}\label{BV}
\begin{tikzpicture}[thick,scale=0.8]
    \begin{scope}[shift={(0,0)}]

        \draw (0.2,5.4) -- (-0.2,4.9) -- (0.6,4.9) -- cycle;
        \node at (0.2,5.1) {\small $y$};

        \draw (0.2,4.6) -- (0.2,4.9);

        \draw (-0.2,4) -- (0.7,4) -- (0.5,4.6) -- (-0.2,4.6)-- cycle;
        \node at (0.2,4.3) {H};
        \draw[dashed] (-0.2,3.7) to [out=0, in=180] (3.9,3.7); 

        \draw (0.2,4) -- (0.2,2.5);

        \filldraw[black] (1,1.5) circle (2.5pt);

        \draw (0.2,2.5) to [out=-90,in=180] (0.9,1.5);

        \draw (1.1,1.5) to [out=0,in=-90] (1.8,2.2);

        \draw (1,1.5) -- (1,0.8);
        \filldraw[black] (1,0.8) circle (2.5pt);

        \draw (1.4,2.2) -- (2.3,2.2) -- (2.1,2.8) -- (1.4,2.8) -- cycle;
        \node at (1.8,2.5) {\small f};

        \draw (1.8,2.8) to [out=90,in=180] (2.5,3.4);

        \node[draw, circle, fill=white, inner sep=2pt] at (2.5,3.4) {};

        \draw (2.5,5) -- (2.1,4.5) -- (2.9,4.5) -- cycle;
        \draw (2.5,4.5) -- (2.5,3.5);
        \node at (2.5,4.7) {\small b};

        \draw (2.6,3.4) to [out=0,in=90] (3.5,2.5) -- (3.5,1.5);
        \draw[dashed] (-0.2,1.3) to [out=0, in=180] (3.9,1.3); 

        \draw (3.5,1.5) -- (3.5,1.2)-- (3.5,1);
        \draw (3.5,0.5) -- (3.1,1) -- (3.9,1) -- cycle;
        \node at (3.5,0.8) {\small b};
    \node[draw, circle, inner sep=1pt] at (0, 0) {$\frac{1}{\sqrt{2^n}}$};
    \end{scope}
\end{tikzpicture}
\end{equation}
\end{center}
where the  $\sqrt{2}\mathrm{X}$ basis represents the $\mathrm{X}$ basis $\{\vert 0 \rangle +\vert 1 \rangle ,  \vert 0 \rangle -\vert 1 \rangle\}$ scaled by $\sqrt{2}$. This structure will be used multiple times in the following text, and we uniformly refer to it as the  $\sqrt{2}\mathrm{X}$ basis.
\end{definition}

The horizontal dashed lines in the diagram divide the algorithm into three stages: preparation, evolution, and measurement, with time proceeding from bottom to top.

\subsection{Correctness of the tensor diagram semantics of the Bernstein-Vazirani}
The diagram semantics of the preparation part is trivial. According to the transformation
\begin{center}
\begin{tikzpicture}[thick,scale=0.8]

    \begin{scope}[shift={(0,0)}]
        \draw (0,2.7) -- (0,1.8);
        \draw (-0.4,1.8) -- (-0.4,1.2) -- (0.5,1.2) -- (0.3,1.8) -- cycle;
        \node at (0,1.5) {$H$};
        \draw (0,1.2) -- (0,0.6);
        \draw (0,-0.1) -- (-0.6,0.6) -- (0.6,0.6) -- cycle;
        \node at (0.05,0.35) {\small $0^{\otimes n}$};
    \end{scope}

    \begin{scope}[shift={(1.5,0)}]
        \draw (0,2.7) -- (0,2.3);
        \draw (-0.4,2.3) -- (-0.4,1.7) -- (0.5,1.7) -- (0.3,2.3) -- cycle;
        \node at (0,2.0) {$Z$};
        \draw (0,1.7) -- (0,1.5);
        \draw (-0.4,1.5) -- (-0.4,0.9) -- (0.5,0.9) -- (0.3,1.5) -- cycle;
        \node at (0,1.2) {$H$};
        \draw (0,0.9) -- (0,0.6);
        \draw (0,-0.1) -- (-0.6,0.6) -- (0.6,0.6) -- cycle;
        \node at (0,0.3) {\small 0};
    \end{scope}

    \node at (2.5, 1.2) {\large $=$};

    \begin{scope}[shift={(3.5,0)}]
        \draw (0.2,2.7) -- (0.2,0.8);
        \filldraw[black] (0.2,0.8) circle (2.5pt); 
        \node[draw, circle, inner sep=1pt] at (-0.2,0.1) {\small $\frac{1}{\sqrt{2^n}}$};
    \end{scope}

    \begin{scope}[shift={(5.0,0)}]
        \draw (0,2.7) -- (0,0.5);
        \draw (0,-0.1) -- (-0.5,0.5) -- (0.5,0.5) -- cycle;
        \node at (0,0.3) {\small $b$};
    \end{scope}

\end{tikzpicture}
\end{center}
we can obtain the bottom part of Figure(\ref{BV}).

As for the unitary evolution part, the role of the oracle $U_{f}$ is actually to keep the state of the first register unchanged, and to perform modulo-2 addition on the result of applying function $f$ to the first register with the second register. So we might as well first make a copy of the first register, apply the function $f$ to the copied new object, so as to ensure that the first register is not changed while obtaining the desired function value $f$:
$$
\vert \psi_{11} \rangle \vert \psi_{12} \rangle \xrightarrow{\Delta_{B} \otimes id} \vert \psi_{11} \rangle \vert \psi_{11} \rangle \vert \psi_{12} \rangle \xrightarrow{id \otimes f \otimes id}  \vert \psi_{11} \rangle \vert f(x) \rangle \vert \psi_{12} \rangle.
$$
Next, we consider how to perform modulo-2 addition on the last two qubit states mentioned above. In fact, the result of function $f$ is either $\vert 0 \rangle$ or $\vert 1 \rangle$ in the computational basis, while at this time $\vert \psi_{12} \rangle = \frac{\vert 0 \rangle - \vert 1 \rangle}{\sqrt{2}}$ is a superposition of the two. If $f(x)=0$, then $\vert \psi_{12} \rangle$ does not change; if $f(x)=1$, then $\vert \psi_{12} \rangle$ flips phase to become $\frac{\vert 1 \rangle - \vert 0 \rangle}{\sqrt{2}}$. Therefore, the reason why the modulo-2 addition here only changes the global phase is that the second bit is in a superposition state with opposite signs belonging to the basis of the function result. So in order not to change the state of the second bit and to transform the result of the first bit into a scalar, we consider using the multiplication in the Frobenius structure $(B, \mu_{W} ,  \eta_{W})$ to which the state of the second bit belongs. Thus, simple calculation shows that if $f(x) = x \cdot s =0$, then

\begin{tikzpicture}[thick, scale=0.7]
    \begin{scope}[shift={(0,3.2)}]
        \draw (0,1.9) -- (0,2.5);
        \draw (-0.1,1.8) to [out=180,in=90] (-0.8,1);
        \draw (0.1,1.8) to [out=0,in=90] (0.8,1);
        \node[draw, circle, inner sep=1.5pt] at (0,1.8) {};
        \draw (-1.2,1) -- (-1.2,0.4) -- (-0.3,0.4) -- (-0.5,1) -- cycle;
        \node at (-0.8,0.7) {$f$};
        \draw (-0.8,0.4) -- (-0.8,0);
        \draw (-0.8,-0.7) -- (-1.3,0) -- (-0.2,0) -- cycle;
        \node at (-0.8,-0.2) {\small $\psi_{11}$};
        \draw (0.8,1) -- (0.8,0);
        \draw (0.8,-0.7) -- (0.3,0) -- (1.3,0) -- cycle;
        \node at (0.8,-0.2) {\small $\psi_{12}$};

        \node at (1.8,0.7) {\Large $=$};

        \begin{scope}[shift={(3.6,0)}]
        \draw (0,1.6) -- (0,2.5);
        \draw (-0.1,1.5) to [out=180,in=90] (-0.8,0.7);
        \draw (0.1,1.5) to [out=0,in=90] (0.8,0.7);
        \node[draw, circle, inner sep=1.5pt] at (0,1.5) {};
        \draw (-0.8,0.7) -- (-0.8,0.3);
        \draw (-0.8,-0.4) -- (-1.3,0.3) -- (-0.2,0.3) -- cycle;
        \node at (-0.8,0.1) {\small $x \cdot s$};
        \draw (0.8,0.7) -- (0.8,0.3);
        \draw (0.8,-0.4) -- (0.3,0.3) -- (1.3,0.3) -- cycle;
        \node at (0.8,0.1) {\small $\psi_{12}$};
        \end{scope}

        \node at (5.4,0.7) {\Large $=$};

        \begin{scope}[shift={(7.2,0)}]
        \draw (0,1.6) -- (0,2.5);
        \draw (-0.1,1.5) to [out=180,in=90] (-0.8,0.7);
        \draw (0.1,1.5) to [out=0,in=90] (0.8,0.7);
        \node[draw, circle, inner sep=1.5pt] at (0,1.5) {};
        \draw (-0.8,0.7) -- (-0.8,0.3);
        \draw (-0.8,-0.4) -- (-1.3,0.3) -- (-0.2,0.3) -- cycle;
        \node at (-0.8,0.1) {\small $0$};
        \draw (0.8,0.7) -- (0.8,0.3);
        \draw (0.8,-0.4) -- (0.3,0.3) -- (1.3,0.3) -- cycle;
        \node at (0.8,0.1) {\small $b$};
        \end{scope}
    \end{scope}


    \node at (-2, 1) {\Large $=$};
    \node at (-1.25, 1) {\Large $\frac{1}{2}$};

    \draw [thick] (-0.5, 2) arc (160:200:3);

    \begin{scope}[shift={(0,0)}]
        \draw (0.96,2) -- (0.96,1.68);
        \draw (0.96,1.12) -- (0.4,1.68) -- (1.52,1.68) -- cycle;
        \node at (0.96,1.48) {\fontsize{6}{12}\selectfont $\sqrt{2}+$};

        \draw (0.4,0.88) -- (0.4,0.56);
        \draw (0.4,1.44) -- (-0.08,0.88) -- (0.88,0.88) -- cycle;
        \node at (0.4,1.08) {\fontsize{6}{12}\selectfont $\sqrt{2}+$};

        \draw (0.4,0.56) -- (0.4,0.4);
        \draw (0.4,-0.16) -- (-0.08,0.4) -- (0.88,0.4) -- cycle;
        \node at (0.4,0.16) {\small 0};

        \draw (1.52,0.88) -- (1.52,0.56);
        \draw (1.52,1.44) -- (1.04,0.88) -- (2.0,0.88) -- cycle;
        \node at (1.52,1.08) {\fontsize{6}{12}\selectfont $\sqrt{2}+$};

        \draw (1.52,0.56) -- (1.52,0.4);
        \draw (1.52,-0.16) -- (1.04,0.4) -- (2.0,0.4) -- cycle;
        \node at (1.52,0.16) {\small b};
    \end{scope}

    \node at (2.8, 0.96) {\large $+$};

    \begin{scope}[shift={(3.6,0)}]
        \draw (0.96,2) -- (0.96,1.68);
        \draw (0.96,1.12) -- (0.4,1.68) -- (1.52,1.68) -- cycle;
        \node at (0.96,1.48) {\fontsize{6}{12}\selectfont $\sqrt{2}-$};

        \draw (0.4,0.88) -- (0.4,0.56);
        \draw (0.4,1.44) -- (-0.08,0.88) -- (0.88,0.88) -- cycle;
        \node at (0.4,1.08) {\fontsize{6}{12}\selectfont $\sqrt{2}-$};

        \draw (0.4,0.56) -- (0.4,0.4);
        \draw (0.4,-0.16) -- (-0.08,0.4) -- (0.88,0.4) -- cycle;
        \node at (0.4,0.16) {\small 0};

        \draw (1.52,0.88) -- (1.52,0.56);
        \draw (1.52,1.44) -- (1.04,0.88) -- (2.0,0.88) -- cycle;
        \node at (1.52,1.08) {\fontsize{6}{12}\selectfont $\sqrt{2}-$};

        \draw (1.52,0.56) -- (1.52,0.4);
        \draw (1.52,-0.16) -- (1.04,0.4) -- (2.0,0.4) -- cycle;
        \node at (1.52,0.16) {\small b};
    \end{scope}

    \draw [thick] (6, 2) arc (20:-20:3);
\end{tikzpicture}

\begin{tikzpicture}[thick, scale=0.7]
    \node at (-2, 1) {\Large $=$};
    \node at (-1.25, 1) {\Large $\sqrt{2}$};
    \draw [thick] (-0.25, 2) arc (160:200:3);

    \begin{scope}[shift={(0,0)}]
        \draw (0.96,2) -- (0.96,1.68);
        \draw (0.96,1.12) -- (0.4,1.68) -- (1.52,1.68) -- cycle;
        \node at (0.96,1.48) {\small $+$};

        \draw (0.4,0.88) -- (0.4,0.56);
        \draw (0.4,1.44) -- (-0.08,0.88) -- (0.88,0.88) -- cycle;
        \node at (0.4,1.08) {\small $+$};

        \draw (0.4,0.56) -- (0.4,0.4);
        \draw (0.4,-0.16) -- (-0.08,0.4) -- (0.88,0.4) -- cycle;
        \node at (0.4,0.16) {\small 0};

        \draw (1.52,0.88) -- (1.52,0.56);
        \draw (1.52,1.44) -- (1.04,0.88) -- (2.0,0.88) -- cycle;
        \node at (1.52,1.08) {\small $+$};

        \draw (1.52,0.56) -- (1.52,0.4);
        \draw (1.52,-0.16) -- (1.04,0.4) -- (2.0,0.4) -- cycle;
        \node at (1.52,0.16) {\small b};
    \end{scope}

    \node at (2.8, 0.96) {\large $+$};

    \begin{scope}[shift={(3.6,0)}]
        \draw (0.96,2) -- (0.96,1.68);
        \draw (0.96,1.12) -- (0.4,1.68) -- (1.52,1.68) -- cycle;
        \node at (0.96,1.48) {\small $-$};

        \draw (0.4,0.88) -- (0.4,0.56);
        \draw (0.4,1.44) -- (-0.08,0.88) -- (0.88,0.88) -- cycle;
        \node at (0.4,1.08) {\small $-$};

        \draw (0.4,0.56) -- (0.4,0.4);
        \draw (0.4,-0.16) -- (-0.08,0.4) -- (0.88,0.4) -- cycle;
        \node at (0.4,0.16) {\small 0};

        \draw (1.52,0.88) -- (1.52,0.56);
        \draw (1.52,1.44) -- (1.04,0.88) -- (2.0,0.88) -- cycle;
        \node at (1.52,1.08) {\small $-$};

        \draw (1.52,0.56) -- (1.52,0.4);
        \draw (1.52,-0.16) -- (1.04,0.4) -- (2.0,0.4) -- cycle;
        \node at (1.52,0.16) {\small b};
    \end{scope}

    \draw [thick] (6, 2) arc (20:-20:3);


    \node at (-2, -1.5) {\Large $=$};
    \node at (-1.25, -1.5) {\Large $\sqrt{2}$};
    \draw [thick] (-0.5, -0.8) arc (160:200:2);

    \begin{scope}[shift={(0,0)}]
        \node at (0.2,-1.5) {\Large $0$};
        \node at (0.9,-1.5) {\Large $+$};
        \node at (1.7,-1.5) {\Large $\frac{1}{\sqrt{2}}$};
        \draw (2.75,-1.5) -- (2.75,-1);
        \draw (2.75,-2.06) -- (3.23,-1.5) -- (2.27,-1.5) -- cycle;
        \node at (2.75,-1.68) {\small $-$};
    \end{scope}
    \draw [thick] (3.5, -0.8) arc (20:-20:2);

        \node at (5, -1.5) {\Large $=$};
        \draw (6.25,-2.06) -- (6.73,-1.5) -- (5.77,-1.5) -- cycle;
        \node at (6.25,-1.68) {\small $-$};
        \draw (6.25,-1.5) -- (6.25,-1);

\end{tikzpicture}

The case for $f(x) = x \cdot s =1$ is similar. So the multiplication $\mu_{W}$ here is exactly the modulo-2 addition we need. Combining this with the previous copying process, we obtain the required oracle $U_{f}$.

The measurement part is also trivial according to the previous diagram operation rules. It is worth noting that in (\ref{BV}) we measure both systems. In fact, this does not affect the algorithm; on the contrary, measuring both systems allows us to understand this algorithm from a better perspective, ultimately revealing the contributions made by each part to ensure the success of the quantum algorithm.

Thus we have completed the full explanation of the tensor diagram for the Bernstein-Vazirani algorithm.

\subsection{Bernstein-Vazirani algorithm from a topological perspective}
Recall again the topological representation of the Bernstein-Vazirani algorithm given in Definition \ref{BVdef}. We have already proven its correctness in the previous text. Now we attempt to reduce it to its simplest form and elucidate the mathematical essence of the algorithm from the topological structure of this minimal representation.

By the copyability of $\sqrt{2}b$, we have
\begin{center}
\begin{tikzpicture}[thick, scale=1]\label{BV7}

    \begin{scope}[shift={(0,0)}]
        \node[draw, circle, inner sep=1.5pt] at (0,1.2) {};
        \draw (0,1.7) -- (0,1.3);
        \draw (-0.1,1.2) to [out=180,in=90] (-0.8,0.5);
        \draw (0.1,1.2) to [out=0,in=90] (0.8,0.5);
        \draw (0.8,0.5) -- (0.8,0.3);
        \draw (-0.8,0.5) -- (-0.8,0.2);
        \draw (0.8,-0.2) -- (0.4,0.3) -- (1.2,0.3) -- cycle;
        \node at (0.8,0.1) {\small $b$};
    \end{scope}

    \node at (2.3, 0.6) {\large $=$};

    \begin{scope}[shift={(4,0)}]
        \draw (0,1.5) -- (0,1.2);
        \draw (0,0.7) -- (-0.4,1.2) -- (0.4,1.2) -- cycle;
        \node at (0,1.0) {\small $b$};
        \draw (0,0.5) -- (-0.4,0) -- (0.4,0) -- cycle;
        \node at (0,0.2) {\small $b$};
        \draw (0,0) -- (0,-0.3);
    \end{scope}
\end{tikzpicture}
\end{center}
Substituting back into (\ref{BV}) yields
\begin{center}
\begin{equation}\label{BV8}
\begin{tikzpicture}[thick, scale=0.7]
    \begin{scope}[shift={(-0.5,-0.5)}]
        \node[draw, circle, inner sep=1pt] at (-0.5,-0.3) {\small $\frac{1}{\sqrt{2^n}}$};
        \draw (1,0) -- (1,0.5);
        \filldraw[black] (1,0.5) circle (2pt);
        \filldraw[black] (1,0) circle (2pt);
        \draw (1,0.5) to [out=180,in=-90] (0,1.2);
        \draw (1,0.5) to [out=0,in=-90] (2,1.2);
        \draw (-0.4,2.3) -- (-0.4,1.7) -- (0.5,1.7) -- (0.3,2.3) -- cycle; \node at (0,2) {$H$};
        \draw (0,2.3) -- (0,3.1);
        \draw (0,3.6) -- (-0.4,3.1) -- (0.4,3.1) -- cycle;
        \node at (0,3.3) {\small $y$};
        \draw (0,1.2) -- (0,1.7);
        \draw (1.6,1.8) -- (1.6,1.2) -- (2.5,1.2) -- (2.3,1.8) -- cycle;
        \node at (2,1.5) {$f$};
        \draw (2,1.8) -- (2,2.1);
        \draw (2,2.1) to [out=90,in=180] (2.9,2.8);
        \draw (3.1,2.8) to [out=0,in=90] (4,2.1);
        \node[draw, circle, fill=white, inner sep=1.5pt] at (3,2.8) {};
        \draw (3,3.7) -- (3.4,3.2) -- (2.6,3.2) -- cycle;
        \node at (3,3.4) {\small $b$};
        \draw (3,2.9) -- (3,3.2);
        \draw (4,2.1) -- (4,0.25);
        \draw (4,-0.25) -- (4.4,0.25) -- (3.6,0.25) -- cycle;
        \node at (4,0) {\small $b$};
    \end{scope}

    \node at (4.5, 1.2) {\large $=$};

    \begin{scope}[shift={(6,0)}]
        \node[draw, circle, inner sep=1pt] at (-0.5,-0.5) {\small $\frac{1}{\sqrt{2^{n-1}}}$};
        \draw (1,0) -- (1,0.5);
        \filldraw[black] (1,0.5) circle (2pt);
        \filldraw[black] (1,0) circle (2pt);
        \draw (1,0.5) to [out=180,in=-90] (0,1.2);
        \draw (1,0.5) to [out=0,in=-90] (2,1.2);
        \draw (-0.4,1.8) -- (-0.4,1.2) -- (0.5,1.2) -- (0.3,1.8) -- cycle; \node at (0,1.5) {$H$};
        \draw (0,1.8) -- (0,2.1);
        \draw (0,2.6) -- (-0.4,2.1) -- (0.4,2.1) -- cycle;
        \node at (0,2.3) {\small $y$};
        \draw (1.6,1.8) -- (1.6,1.2) -- (2.5,1.2) -- (2.3,1.8) -- cycle;
        \node at (2,1.5) {$f$};
        \draw (2,1.8) -- (2,2.1);
        \draw (2,2.6) -- (2.4,2.1) -- (1.6,2.1) -- cycle;
        \node at (2,2.3) {\small $b$};
    \end{scope}
\end{tikzpicture}
\end{equation}
\end{center}
If we view $y \circ H$ as a whole measurement, it can be further simplified to
\begin{center}
\begin{equation}\label{BV9}
\begin{tikzpicture}[thick, scale=0.9]
    \begin{scope}[shift={(0,0)}]
        \node[draw, circle, inner sep=0.5pt] at (-1,0) {\small $\frac{1}{\sqrt{2^{n-1}}}$};
        \draw (0.5,0.8) -- (0.5,1.2); \filldraw (0.5,0.8) circle (2pt);
        \draw (2,0.8) -- (2,1.2); \filldraw (2,0.8) circle (2pt);
        \draw[dashed] (0.5,0.8) to [out=-30, in=210] (2,0.8); 
        \draw (0.5,1.2) -- (0.5,1.5);
        \draw (0.5,2.15) -- (-0.1,1.5) -- (1.1,1.5) -- cycle; \node at (0.5,1.7) {\small $M$};
        \draw (1.6,1.8) -- (1.6,1.2) -- (2.5,1.2) -- (2.3,1.8) -- cycle; \node at (2,1.5) {$f$};
        \draw (2,1.8) -- (2,2.1);
        \draw (2,2.6) -- (1.6,2.1) -- (2.4,2.1) -- cycle; \node at (2,2.3) {\small $b$};
    \end{scope}
\end{tikzpicture}
\end{equation}
\end{center}
where the dashed line indicates that the two are in an entangled state rather than a product state, i.e.
\begin{center}
\begin{tikzpicture}[thick,scale=0.8]
    \filldraw (0,0) circle (2.5pt);
    \filldraw (1.5,0) circle (2.5pt);
    \draw (0,0) -- (0,0.6);
    \draw (1.5,0) -- (1.5,0.6);
    \draw[dashed] (0,0) to [out=-30, in=210] (1.5,0);

    \node at (2.7, 0.3) {\Large $=$};
    \node at (4.5, 0.3) {\large $\displaystyle \sum_{x} \vert x\rangle \vert x\rangle$};
\end{tikzpicture}
\end{center}
Thus we have obtained the simplest form (\ref{BV9}) of the tensor diagram semantics of the Bernstein-Vazirani algorithm after deformation.

Here we view the measurement $y \circ H$ as a whole measurement $M$, because for the simplest diagram semantics, we are not concerned with the specific value of a certain morphism or state, but rather focus on the overall ``change" which is precisely the important idea of category theory.

Obviously, we concern how the various components in the simplest form of the tensor diagram semantics correspond to the original quantum circuit, given that the two appear substantially different at first glance. Without any doubt, the first object space $(\mathbb{C}^2)^{\otimes n}$ corresponds to the first $n$ bits of the first register in the quantum circuit. Then why is there no single qubit of the second register in the diagram? In fact, this stems from the copyability of $\sqrt{2}b$ in the simplification process (\ref{BV7}) and our measurement of the second register, which causes the evolution process of the second object to become applying the effect $\langle b \vert$ after preparing the state $\vert b \rangle$. So from the overall perspective is a morphism $I \rightarrow I$, which is an empty diagram.
\begin{center}
\begin{tikzpicture}[thick]

    \draw (0,0.3) -- (0,0.7);
    \draw (0,1.2) -- (-0.4,0.7) -- (0.4,0.7) -- cycle;
    \node at (0,0.9) {\small $b$};

    \draw (0,0.3) -- (0,-0.3);

    \draw (0,-0.8) -- (-0.4,-0.3) -- (0.4,-0.3) -- cycle;
    \node at (0,-0.5) {\small $b$};

    \node at (1.5, 0.2) {\large $=$};

    \node at (2.5, 0) {};

\end{tikzpicture}
\end{center}
From another perspective, this also shows the second register is an introduced auxiliary qubit. Unlike the quantum circuit where only the first qubit is measured, the reason we measure both systems in the tensor diagram semantics is to more clearly demonstrate the auxiliary nature of the second register.

Although the effect $\langle y \vert$ and the function $f$ appear to be distributed across two distinct objects, tracing back along the time axis reveals that their initial states are in fact entangled—that is, their evolutionary origins are in a ``synchronized" state. Consequently, measuring the first object yields relevant information about the second object. The fundamental reason why measuring only the first object suffices to obtain complete information about the function $f$ is as follows: in the Bernstein-Vazirani algorithm, the action of the oracle not only leaves the relative phase of the auxiliary qubit unchanged, but also preserves all outcomes of the function $f$ in the global phase. This is precisely the manifestation of phase kickback: although the controlled-phase gate may employ the top qubit as the control terminal, the phase shift nonetheless affects the top qubit itself.
\begin{center}
\begin{tikzpicture}[thick, scale=0.8]
    \draw (0,0) rectangle (2,2);
    \node at (1,1) {\Large $U_f$};
    \draw (-0.8, 1.7) -- (0,1.7) node[pos=0, left] {\small $x$};
    \draw (2, 1.7) -- (2.8,1.7) node[pos=1, right] {\small $x$};
    \draw (-0.8, 0.3) -- (0,0.3) node[pos=0, left] {\small $y$};
    \draw (2, 0.3) -- (2.8,0.3) node[pos=1, right] {\small $y \oplus f(x)$};
    \node at (-1.2, 2.2) {\small $\text{control}$};

    \node at (1, -1) {$\updownarrow$};

    \node at (1, -2) {
        $\displaystyle \sum_{x} \frac{|x\rangle}{\sqrt{2^n}} \left[ \frac{|0\rangle - |1\rangle}{\sqrt{2}} \right]
        \xrightarrow{U_f}
        \sum_{x} \frac{(-1)^{s \cdot x}|x\rangle}{\sqrt{2^n}} \left[ \frac{|0\rangle - |1\rangle}{\sqrt{2}} \right]$
    };
\end{tikzpicture}
\end{center}

Hence, the measurement probability is non-zero if and only if $y=s$. By measuring the first object alone, one obtains complete information about the string $s$. The minimal tensor diagram semantics thus provides a lucid and intuitive explanation of the essential nature of the oracle in the Bernstein-Vazirani algorithm and of the algorithm's underlying idea—namely, the properties of quantum entanglement and quantum superposition.

\section{Topological Simon}
Building upon the tensor diagram semantics of the Bernstein-Vazirani algorithm presented in the previous section, we now analyze another classic protocol: Simon's algorithm..

\subsection{Introduction}
Simon's algorithm was proposed by Daniel R. Simon in 1994. It emerged at a critical turning point in the development of quantum computing—before this, although the Deutsch-Jozsa algorithm and the Bernstein-Vazirani algorithm had proven the query complexity advantages of quantum computing, the problems they addressed were still regarded as highly artificial theoretical constructions. However, Simon's algorithm can achieve a balance between ``naturalness" and ``structurality". It was the first to apply quantum acceleration to a structured problem with clear cryptographic relevance, building a key bridge for quantum computing to move from theoretical singularity to practical value. It requires only $O(n)$ quantum queries to determine the hidden period, while the optimal classical probabilistic algorithm requires an exponential number of queries. This exponential separation not only demonstrates the computational power of quantum parallelism and quantum entanglement, but more importantly, its problem structure directly inspired Peter Shor's quantum algorithms for integer factorization and discrete logarithms proposed at the end of 1994.

Simon's algorithm is also a promise problem related to a function. Suppose we have a hidden binary string $s\in \{0, 1\}^{n}$, and the function $f: \{0, 1\}^{n} \rightarrow \{0, 1\}^{n}$ satisfies $f(x)=f(y)$ if and only if $y = x \oplus s$. The Simon problem is also to find the hidden string $s$. Its quantum circuit is as follows.
\begin{center}
\begin{tikzpicture}[thick, scale=0.9]

    \node at (-1.8, 1.5) {$|0\rangle$};
    \node at (-1.8, 0.2) {$|0\rangle$};

    \draw (-1.5, 1.5) -- (-0.6, 1.5);
    \draw (-1.3, 1.35) -- (-1.0, 1.65); 
    \node at (-1.15, 1.8) {\small $n$};

    \draw (-0.6, 1.1) rectangle (0.4, 1.9);
    \node at (-0.1, 1.5) {$H^{\otimes n}$};

    \draw (0.4, 1.5) -- (1.5, 1.5);

    \draw (-1.5, 0.2) -- (1.5, 0.2);
    \draw (-1.3, 0.05) -- (-1.0, 0.35); 
    \node at (-1.15, 0.5) {\small $n$};

    \draw (1.5, -0.4) rectangle (3.5, 2.1);
    \node at (2.5, 0.85) {\Large $U_f$};

    \draw (3.5, 1.5) -- (4.2, 1.5);
    \draw (4.2, 1.1) rectangle (5.2, 1.9);
    \node at (4.7, 1.5) {$H^{\otimes n}$};
    \draw (5.2, 1.5) -- (6.2, 1.5);

    \draw (6.2, 1.1) rectangle (7.2, 1.9);
    \draw (6.4, 1.3) arc (150:30:0.4);
    \draw[->, >=stealth] (6.7, 1.2) -- (7.0, 1.7);

    \draw (3.5, 0.2) -- (6.2, 0.2);
    \draw (6.2, -0.2) rectangle (7.2, 0.6);
    \draw (6.4, 0) arc (150:30:0.4);
    \draw[->, >=stealth] (6.7, -0.1) -- (7.0, 0.4);

    \begin{scope}[every node/.style={inner sep=1pt}]
        \node at (-1.1, -1.2) (p0) {$|\psi_0\rangle$};
        \draw[->] (p0) -- (-1.1, -0.7);

        \node at (0.85, -1.2) (p1) {$|\psi_1\rangle$};
        \draw[->] (p1) -- (0.85, -0.7);

        \node at (3.8, -1.2) (p2) {$|\psi_2\rangle$};
        \draw[->] (p2) -- (3.8, -0.7);

        \node at (5.7, -1.2) (p3) {$|\psi_3\rangle$};
        \draw[->] (p3) -- (5.7, -0.7);
    \end{scope}
\end{tikzpicture}
\end{center}
\begin{definition}[Topological Simon] \label{simondef}
Let $A=(\mathbb{C}^2)^{\otimes n}$, with its Frobenius structure $(A, \mu_{B} ,  \eta_{B})$ taken as the computational basis. For an $n$-bit string $s$ and a given function $A \xrightarrow{f} B$, the tensor diagram semantics of Simon's algorithm in $\mathbf{FHilb}$ is as follows:
\begin{center}
\begin{equation}\label{def4.1}
\begin{tikzpicture}[thick, scale=0.9]

    \draw (0.3, 4.7) -- (-0.1, 4.1) -- (0.7, 4.1) -- cycle;
    \node at (0.3, 4.25) {\small $y_0$};
    \draw (0.3, 3.9) -- (0.3, 4.1);

    \draw (2.4, 4.2) -- (2, 3.6) -- (2.8, 3.6) -- cycle;
    \node at (2.4, 3.75) {\small $y_1$};
    \draw (2.4, 3.4) -- (2.4,3.6);

    \draw (-0.1, 3.9) -- (-0.1, 3.4) -- (0.8, 3.4) -- (0.6, 3.9) -- cycle;
    \node at (0.3, 3.6) {$H$};
    \draw (0.3, 1.4) -- (0.3, 3.4); 

    \draw (1.3, 2.1) -- (1.3, 1.6) -- (2.2, 1.6) -- (2, 2.1) -- cycle;
    \node at (1.7, 1.8) {$f$};
    \draw (1.7, 1.4) -- (1.7, 1.6); 
    \draw (1.7, 2.1) -- (1.7, 2.3); 

    \draw (0.3, 1.4) to [out=-90, in=180] (1, 0.8);
    \draw (1.7, 1.4) to [out=-90, in=0] (1, 0.8);
    \filldraw (1, 0.8) circle (2pt); 
    \draw (1, 0.8) -- (1, 0.2); 
    \filldraw (1, 0.2) circle (2pt); 

    \draw (3.1, 2.3) -- (3.1, 0.3);

    \node[draw, circle, inner sep=1pt] at (0, 0) {$\frac{1}{\sqrt{2^n}}$};

        \draw (3.1, -0.1) -- (2.7, 0.3) -- (3.5, 0.3) -- cycle;
        \node at (3.1, 0.1) {\fontsize{6}{12}\selectfont $0^{\otimes n}$};

    \draw (1.7, 2.3) to [out=90, in=180] (2.4, 2.9);
    \draw (2.4, 2.9) to [out=0, in=90] (3.1, 2.3);
    \draw (2.4, 2.9) -- (2.4, 3.4); 
    \draw[fill=white] (2.4, 2.9) circle (2pt); 
    \draw[dashed, gray] (-1, 3.1) -- (4, 3.1);
    \draw[dashed, gray] (-1, 0.6) -- (4, 0.6);

\end{tikzpicture}
\end{equation}
\end{center}
\end{definition}

\subsection{Correctness of the tensor diagram semantics of Simon algorithm}
The unitary gate $U_{f}:\vert x, y \rangle \rightarrow \vert x, y \oplus f(x) \rangle$ in the quantum circuit is identical to that of the Bernstein-Vazirani algorithm, but here the role of the function $f$ is not merely reflected in the amplitude of the quantum state, but directly changes the second register. Since the state of the first register does not change, we still first make a copy of $\vert \psi_{11} \rangle$, and then apply the function $f$ to the copied new system. Performing bitwise modulo-2 addition with ${\vert 0 \rangle}^{\otimes n}$ means that the result completely preserves the state $ \vert f(x) \rangle$. 
The topological representation of this transformation can be given as follows:
\begin{center}
\begin{tikzpicture}[thick, scale=0.8]

    \begin{scope}[shift={(0,0)}]
        \draw (1, 3.2) -- (1, 2.5);
        \node at (1, 3.5) {$|\varphi\rangle$};

        \draw (0.3, 1.3) rectangle (1.7, 2.5);
        \node at (1, 1.9) {\Large $\oplus$};

        \draw (0.6, 1.3) -- (0.6, 0.5);

        \begin{scope}[shift={(0.7, 0.4)}]
            \node[anchor=north] at (-0.1, 0) {\small $|\varphi\rangle$};
            \node[anchor=north] at (0.9, 0) {\small $|0\rangle^{\otimes n}$};
        \end{scope}
            \draw (1.4, 1.3) -- (1.4, 0.5);

    \end{scope}

    \node at (2.6, 1.9) {\large $=$};

    \begin{scope}[shift={(3.5,0)}]
        \draw (1, 3.2) -- (1, 2.4);
        \node at (1, 3.5) {$|\varphi\rangle$};
        \draw (1, 2.4) to [out=180, in=90] (0, 1.5);
        \draw (0, 1.5) -- (0, 0.5);
        \node[anchor=north] at (0, 0.5) {$|\varphi\rangle$};
        \draw (1, 2.4) to [out=0, in=90] (2, 1.5);
        \draw (2, 1.5) -- (2, 0.4);

        \draw[fill=white] (1, 2.4) circle (2pt);

        \draw (2, -0.1) -- (1.6, 0.4) -- (2.4, 0.4) -- cycle;
        \node at (2, 0.2) {\small $0$};
    \end{scope}

\end{tikzpicture}
\end{center}
Thus far, we have once again simulated the oracle $U_{f}$ required in the quantum circuit.

\subsection{Simon algorithm from a topological perspective}
By scalar calculation and the unitality of the monoid, the topological representation of Simon's algorithm given in Definition \ref{simondef} can be deformed into:
\begin{center}
\begin{tikzpicture}[thick, scale=0.75]

    \begin{scope}[shift={(0,0)}]
        \node[draw, circle, inner sep=1pt] at (-0.5,-0.3) {\small $\frac{1}{\sqrt{2^n}}$};
        \draw (1,0) -- (1,0.5);
        \filldraw[black] (1,0.5) circle (2pt);
        \filldraw[black] (1,0) circle (2pt);
        \draw (1,0.5) to [out=180,in=-90] (0,1.2);
        \draw (1,0.5) to [out=0,in=-90] (2,1.2);
        \draw (-0.4,2.3) -- (-0.4,1.7) -- (0.5,1.7) -- (0.3,2.3) -- cycle; \node at (0,2) {$H$};
        \draw (0,2.3) -- (0,3.1);
        \draw (0,3.6) -- (-0.4,3.1) -- (0.4,3.1) -- cycle;
        \node at (0,3.3) {\small $y_0$};
        \draw (0,1.2) -- (0,1.7);
        \draw (1.6,1.8) -- (1.6,1.2) -- (2.5,1.2) -- (2.3,1.8) -- cycle;
        \node at (2,1.5) {$f$};
        \draw (2,1.8) -- (2,2.1);
        \draw (2,2.1) to [out=90,in=180] (2.9,2.8);
        \draw (3.1,2.8) to [out=0,in=90] (4,2.1);
        \draw[fill=white] (3,2.8) circle (2pt);
        \draw (3,3.7) -- (3.4,3.2) -- (2.6,3.2) -- cycle;
        \node at (3,3.4) {\small $y_1$};
        \draw (3,2.9) -- (3,3.2);
        \draw (4,2.1) --  (4,0.25);
        \draw (4,-0.25) -- (4.4,0.25) -- (3.6,0.25) -- cycle;
        \node at (4,0.05) {\small $0$};

    \end{scope}

    \node at (5.2, 1.5) {\large $=$};

    \begin{scope}[shift={(7,0.5)}]
        \node[draw, circle, inner sep=1pt] at (-0.4,-0.25) {\small $\frac{1}{\sqrt{2^n}}$};
        \draw (1,0) -- (1,0.5);
        \filldraw[black] (1,0.5) circle (2pt);
        \filldraw[black] (1,0) circle (2pt);
        \draw (1,0.5) to [out=180,in=-90] (0,1.2);
        \draw (1,0.5) to [out=0,in=-90] (2,1.2);
        \draw (-0.4,1.8) -- (-0.4,1.2) -- (0.5,1.2) -- (0.3,1.8) -- cycle; \node at (0,1.5) {$H$};
        \draw (0,1.8) -- (0,2.1);
        \draw (0,2.6) -- (-0.4,2.1) -- (0.4,2.1) -- cycle;
        \node at (0,2.3) {\small $y_0$};
        \draw (1.6,1.8) -- (1.6,1.2) -- (2.5,1.2) -- (2.3,1.8) -- cycle;
        \node at (2,1.5) {$f$};
        \draw (2,1.8) -- (2,2.1);
        \draw (2,2.6) -- (2.4,2.1) -- (1.6,2.1) -- cycle;
        \node at (2,2.3) {\small $y_1$};
    \end{scope}

\end{tikzpicture}
\end{center}
Further using dashed lines to indicate that the two uniform superposition states are in an entangled state, we can obtain the simplest form of the topological representation of Simon's algorithm:
\begin{center}
\begin{tikzpicture}[thick]
    \begin{scope}[shift={(10.5,0)}]
        \node[draw, circle, inner sep=0.5pt] at (-1,0.8) {\small $\frac{1}{\sqrt{2^n}}$};
        \draw (0.5,0.8) -- (0.5,1.2); \filldraw (0.5,0.8) circle (2pt);
        \draw (2,0.8) -- (2,1.2); \filldraw (2,0.8) circle (2pt);
        \draw [dashed] (0.5,0.8) to [out=-30, in=210] (2,0.8); 
        \draw (0.5,1.2) -- (0.5,1.5);
        \draw (0.5,2) -- (0.1,1.5) -- (0.9,1.5) -- cycle; \node at (0.5,1.65) {\small $y_0$};
        \draw (1.6,1.8) -- (1.6,1.2) -- (2.5,1.2) -- (2.3,1.8) -- cycle; \node at (2,1.5) {$f$};
        \draw (2,1.8) -- (2,2.1);
        \draw (2,2.6) -- (1.6,2.1) -- (2.4,2.1) -- cycle; \node at (2,2.25) {\small $y_1$};
    \end{scope}

\end{tikzpicture}
\end{center}
Re-examining Simon's algorithm from the perspective of diagram semantics reveals a striking superficial resemblance to the minimal form of the Bernstein-Vazirani algorithm we previously presented; however, their underlying topological mechanisms differ fundamentally. In the minimal diagram semantics of Simon's algorithm, the first object likewise represents the first register, but the crucial distinction lies in the fact that the second object here directly represents the second register itself. This can be interpreted from two complementary angles: at the level of tensor diagram semantics, it originates from the unitality of the monoid in the simplification process; computationally, it stems from the fact that $(\vert 0 \rangle)^{\otimes n}$ serves as the identity element under modulo-2 addition.

Secondly, the disappearance of the third object during the transformation process arises from distinct mechanisms in each case: in one instance, it stems from measurement following state preparation; in the other, it derives from the construction of the black box involving the unitality of the Frobenius structure. Thus, despite their apparent structural similarity, the algorithmic functionality and topological essence remain fundamentally different.

\section{Qutrit-adapted generalized topological Deutsch-Jozsa and generalized topological single-shot Grover}

\subsection{Deutsch-Jozsa algorithm for qutrit}

While the Deutsch–Jozsa algorithm is well-established for the qubit case, its generalization to higher dimensions exists; in this work, we focus specifically on the qutrit case. To clarify these concepts, we first define the notions of constant and balanced functions.
\begin{definition}
An $r$-qutrit multi-valued function of the form $f: \{0,1,2\}^{r} \rightarrow \{0,1,2\}$ is \textbf{constant} when $f(x) = f(y)$ for all $x, y \in \{0,1,2\}^{r}$, and is \textbf{balanced} when an equal number of the $3^{r}$ domain values, namely $3^{r-1}$, is mapped to each of the $3$ elements in the codomain.
\end{definition}

Contrary to the qubit case , the Hadmard gate $H$ is replaced by quantum Fourier transform
$\mathcal{F}_{3} = \frac{1}{\sqrt{3}} \sum_{j=0}^{2} \sum_{k=0}^{2} e^{i2\pi j k / 3} | j \rangle \langle k |$. The following circuit implements the Deutsch-Jozsa algorithm for a qutrit.

\[
\tikzset{every picture/.style={line width=0.75pt}} 
\begin{tikzpicture}[x=0.555pt,y=0.55pt,yscale=-1,xscale=1]
\draw    (152,99) -- (229.92,99.27) ;
\draw   (230,86.26) -- (275.67,86.26) -- (275.67,112.26) -- (230,112.26) -- cycle ;
\draw    (274,99) -- (298.67,99.26) ;
\draw   (299,86.27) -- (379.26,86.27) -- (379.26,230.27) -- (299,230.27) -- cycle ;
\draw   (412,87.26) -- (457.67,87.26) -- (457.67,113.26) -- (412,113.26) -- cycle ;
\draw    (456,100) -- (480.67,100.26) ;
\draw    (380,100) -- (410.26,100.27) ;
\draw   (481,87.26) -- (526.67,87.26) -- (526.67,113.26) -- (481,113.26) -- cycle ;
\draw    (485,107) .. controls (496.26,95.27) and (510.26,99.27) .. (518.26,106.27) ;
\draw    (511.26,91.27) -- (497.26,107.27) ;
\draw    (153,212) -- (230.92,212.27) ;
\draw   (231,199.26) -- (276.67,199.26) -- (276.67,225.26) -- (231,225.26) -- cycle ;
\draw    (275,212) -- (299.67,212.26) ;
\draw    (378.26,214.27) -- (407.92,214.27) ;
\draw   (409,199.26) -- (454.67,199.26) -- (454.67,225.26) -- (409,225.26) -- cycle ;
\draw    (453,212) -- (477.67,212.26) ;
\draw    (172,86) -- (163.26,111.27) ;
\draw  [dash pattern={on 4.5pt off 4.5pt}]  (198.26,83.27) -- (199.26,232.27) ;
\draw  [dash pattern={on 4.5pt off 4.5pt}]  (286.26,85.27) -- (287.26,234.27) ;
\draw  [dash pattern={on 4.5pt off 4.5pt}]  (394.26,88.27) -- (395.26,237.27) ;
\draw  [dash pattern={on 4.5pt off 4.5pt}]  (470.26,87.27) -- (471.26,236.27) ;
\draw (115,91) node [anchor=north west][inner sep=0.75pt]   [align=left] {$\displaystyle \ket{0}$};
\draw (236,88) node [anchor=north west][inner sep=0.75pt]    {$\mathcal{F}_{3}^{\otimes r}$};
\draw (328,155) node [anchor=north west][inner sep=0.75pt]   [align=left] {$\displaystyle U_{f}$};
\draw (418,88) node [anchor=north west][inner sep=0.75pt]    {$\mathcal{F}_{3}^{\otimes r}$};
\draw (237,203.4) node [anchor=north west][inner sep=0.75pt]    {$\mathcal{F}_{3}$};
\draw (416,203.4) node [anchor=north west][inner sep=0.75pt]    {$\mathcal{F}_{3}$};
\draw (118,202) node [anchor=north west][inner sep=0.75pt]   [align=left] {$\displaystyle \ket{1}$};
\draw (493,202) node [anchor=north west][inner sep=0.75pt]   [align=left] {$\displaystyle \ket{2}$};
\draw (175,79) node [anchor=north west][inner sep=0.75pt]   [align=left] {$\displaystyle r$};
\draw (182,249) node [anchor=north west][inner sep=0.75pt]   [align=left] {$\displaystyle \ket{\psi _{1}}$};
\draw (270,250) node [anchor=north west][inner sep=0.75pt]   [align=left] {$\displaystyle \ket{\psi _{2}}$};
\draw (376,250) node [anchor=north west][inner sep=0.75pt]   [align=left] {$\displaystyle \ket{\psi _{3}}$};
\draw (457,250) node [anchor=north west][inner sep=0.75pt]   [align=left] {$\displaystyle \ket{\psi _{4}}$};
\end{tikzpicture}
\]

According to the quantum circuit, we write down the middle state.
\[
\hspace{-2.5em}
\begin{aligned}
\ket{\psi_{1}} &= |0\rangle^{\otimes r}\ket{1} \\
\xrightarrow{\mathcal{F}_{3}^{\otimes r+1}}
\ket{\psi_{2}} &= \frac{1}{\sqrt{3^{r}}}\sum_{x=0}^{3^{r}-1}\ket{x}\otimes
\frac{1}{\sqrt{3}}\sum_{y=0}^{2}e^{i2\pi y/3}\ket{y} \\
\xrightarrow{U_{f}}
\ket{\psi_{3}} &= \frac{1}{\sqrt{3^{r}}}\sum_{x=0}^{n^{r}-1}\ket{x}\otimes
\frac{1}{\sqrt{3}}\sum_{y=0}^{2}e^{i2\pi y/3}\ket{y\oplus f(x)}
\\
 &= \frac{1}{\sqrt{3^{r}}} \sum_{x=0}^{3^{r}-1} \ket{x} \otimes
                \frac{1}{\sqrt{3}} \sum_{y=0}^{2} e^{i2\pi[y-f(x)]/3} \ket{y} \\
              &= \frac{1}{\sqrt{3^{r}}} e^{-i2\pi f(x)/3} \sum_{x=0}^{3^{r}-1} \ket{x} \otimes
                \frac{1}{\sqrt{3}} \sum_{y=0}^{2} e^{i2\pi y/3} \ket{y}
\end{aligned}
\]

When $f$ is constant, $ e^{-i2\pi f(x)/3}$ is a global phase factor

\[
\begin{aligned}
\ket{\psi_{3}}
              =& \frac{1}{\sqrt{3^{r}}} e^{-i2\pi f(x)/3} \sum_{x=0}^{3^{r}-1} \ket{x} \otimes
                \frac{1}{\sqrt{3}} \sum_{y=0}^{2} e^{i2\pi y/3} \ket{y} \\
\xrightarrow{\mathcal{F}_{3}^{\otimes r+1}}
\ket{\psi_{4}} =&
\frac{1}{3^{r}} e^{-i2\pi f(x)/3} \sum_{j=0}^{3^{r}-1} \sum_{k=0}^{3^{r}-1} \sum_{x=0}^{3^{r}-1} e^{i2\pi jk/3} |j\rangle\langle k|x\rangle\\
&\otimes \sum_{j=0}^{2} \sum_{k=0}^{2} \sum_{y=0}^{2} e^{i2\pi(j \oplus 1)k/3} |j\rangle\langle k\ket{y} \\
\end{aligned}
\]

For a non-zero integer $j$, we have
$
\sum_{k=0}^{2} (e^{i2\pi /3})^{j k} = 0
$.

Thus,
\[
\begin{aligned}
\ket{\psi_{4}} =& \frac{1}{3^{r}} e^{-i2\pi f(x)/3} \sum_{j=0}^{3^{r}-1} \sum_{k=0}^{3^{r}-1} e^{i2\pi jk/3} \ket{j} \\
&\otimes
\frac{1}{3} \sum_{j=0}^{2} \sum_{k=0}^{2} e^{i2\pi(j \oplus 1)k/3} \ket{j}\\
=& e^{-i2\pi f(x)/3}|0\rangle^{\otimes r} |2\rangle
\end{aligned}
\]

And we observe $|0\rangle^{\otimes r}$ with probability $1$.

When $f$ is balanced, we consider the first $r$ qutrits:
\[
\begin{aligned}
|\psi_{3}\rangle_r &= \frac{1}{\sqrt{3^{r}}} \sum_{x=0}^{3^{r}-1} e^{-i2\pi f(x)/3} |x\rangle \\
  \xrightarrow{\mathcal{F}_{3}^{\otimes r}}  |\psi_{4}\rangle_r
& =  \displaystyle \frac{1}{3^{r}} \sum_{j=0}^{3^{r}-1} \sum_{k=0}^{3^{r}-1} \sum_{x=0}^{3^{r}-1} e^{i2\pi jk/3} e^{-i2\pi f(x)/3} |j\rangle \langle k | x \rangle \\
& =  \displaystyle \frac{1}{3^{r}} \sum_{j=0}^{3^{r}-1} \sum_{x=0}^{3^{r}-1} e^{i2\pi (jx - f(x))/3} |j\rangle.
\end{aligned}
\]

Note that $\sum_{x=0}^{3^{r}-1} e^{i2\pi ( - f(x))/3} |0\rangle=0$ holds for the balanced function, where the phases cancel. Consequently, the state $|0\rangle^{\otimes r}$ is observed with zero probability after measurement. In contrast, for a constant function, we do
 observe $|0\rangle^{\otimes r}$
with probability $1$.

\subsection{Topological Deutsch-Jozsa}
In Sections 4.2 and 4.3, we present new results concerning the topological Deutsch-Jozsa and single-shot Grover algorithms for multi-valued quantum systems, originally introduced by \cite{Vic13}. A complete set of simplification rules for this framework is established in \cite{HV19}, while a more concrete formulation can be found in \cite{Ran14}.

To say that a function $f: S \to \{0,1,2\}$ (where $\vert{}S\vert{} = 3^n$) is constant is equivalent to saying that it factors through a singleton set:

\[
\tikzset{every picture/.style={line width=0.75pt}} 

\]
The function $f$ is defined by a function $x$ with $f(s)=x(1)$ for all $s\in S$.
It has the following graphical expression.
\[
\tikzset{every picture/.style={line width=0.75pt}} 
%
\]

When $f$ is balanced, we have
\[
\left(\begin{array}{ccc} 1 & \omega^2 & \omega \end{array}\right) \circ f \circ \sum_{s \in S} |s\rangle = 0,
\]
 where $f \circ \sum_{s \in S} |s\rangle := \sum_{s \in S} |f(s)\rangle$, $|s\rangle$ is regarded as a column vector and $\omega = e^{2\pi i / 3}$. Furthermore, it can be reformulated as an effect
$ \langle 0 |+\omega^2\langle 1 |+(\omega^2)^2\langle 2 | .$ This yields the following graphical expression:

\[
\tikzset{every picture/.style={line width=0.75pt}} 
%
\]
At the heart of the topological Deutsch–Jozsa algorithm lies the unitary operator $U_f$, which is identical to that of the Bernstein-Vazirani algorithm:

\[
\tikzset{every picture/.style={line width=0.75pt}} 
%
\]
To execute the algorithm, we begin by initializing the input state:

\[
\tikzset{every picture/.style={line width=0.75pt}} 
%
\]
Following the implementation of the protocol and prior to measurement, the system yields the following diagrammatic expression:

\[
\tikzset{every picture/.style={line width=0.75pt}} 
%
\]
By applying the copy rule, we obtain the following expression:

\[
\tikzset{every picture/.style={line width=0.75pt}} 
%
\]
Next, we examine the behavior of the function $f$.

Under the condition that $f$ is constant, we have:

\[
\tikzset{every picture/.style={line width=0.75pt}} 
%
\]
Since $\omega^2 \circ x(1) = \omega^i$ for some $i \in \{0, 1, 2\}$, the expression can be simplified as follows:

\[
\tikzset{every picture/.style={line width=0.75pt}} 
%
\]
Therefore, the projective measurement on the first $r$ qutrits onto the first uniform superposition state yields a deterministic success.

Conversely, when $f$ is balanced, we project the first $r$ qutrits onto the uniform superposition state:

\[
\tikzset{every picture/.style={line width=0.75pt}} 
%
\]
The following result holds:

\[
\tikzset{every picture/.style={line width=0.75pt}} 

\]
Since $\left(\begin{array}{ccc} 1 & \omega^2 & \omega \end{array}\right) \circ f \circ \sum_{s \in S} |s\rangle = 0$, the overall expression vanishes. Consequently, there is a zero probability of finding the first $r$ qutrits in the uniform superposition state upon measurement.

\subsection{The generalized topological single-shot Grover }
 The  generalized single-shot Grover algorithm is used to search a set $S$ for marked elements, defined in terms of an indicator function $S \xrightarrow{f}G$. Here we consider the case $G=\mathbb{Z}_n$ , especially the $\mathbb{Z}_3$ case.
\begin{definition}
Given an irreducible representation
$\rho$ of $G$ , we define the element $s \in S $ to be balanced if the following holds:
\[
\rho(f(s))=\frac{2}{|S|}\sum_{t\in S}\rho(f(t))
\]
\end{definition}

Notice that an effect can be interpreted as a type of representation. Below, we present the topological formulation of the generalized single-shot Grover algorithm for the $\mathbb{Z}_3$ case:

\[
\tikzset{every picture/.style={line width=0.75pt}} 
\begin{tikzpicture}[x=0.75pt,y=0.75pt,yscale=-1,xscale=1]
\draw    (304,172.76) -- (304,215.76) ;
\draw [shift={(304,215.76)}, rotate = 90] [color={rgb, 255:red, 0; green, 0; blue, 0 }  ][fill={rgb, 255:red, 0; green, 0; blue, 0 }  ][line width=0.75]      (0, 0) circle [x radius= 3.35, y radius= 3.35]   ;
\draw    (367.67,134.76) -- (367.67,217.76) ;
\draw    (284,144.76) .. controls (285.76,154.79) and (290.76,170.79) .. (304,172.76) ;
\draw [shift={(304,172.76)}, rotate = 8.43] [color={rgb, 255:red, 0; green, 0; blue, 0 }  ][fill={rgb, 255:red, 0; green, 0; blue, 0 }  ][line width=0.75]      (0, 0) circle [x radius= 3.35, y radius= 3.35]   ;
\draw    (304,172.76) .. controls (315.76,168.79) and (320.76,160.79) .. (324,144.76) ;
\draw    (324,124.76) .. controls (326.6,110.62) and (330.32,100.91) .. (341.9,95.63) ;
\draw [shift={(344,94.76)}, rotate = 339.17] [color={rgb, 255:red, 0; green, 0; blue, 0 }  ][line width=0.75]      (0, 0) circle [x radius= 3.35, y radius= 3.35]   ;
\draw    (347.67,94.76) .. controls (363.42,100.79) and (365.42,111.79) .. (367.67,134.76) ;
\draw    (344,47.42) -- (344,66.46) -- (344,90.42) ;
\draw   (367.11,236.76) -- (350.55,217.83) -- (383.67,217.83) -- cycle ;
\draw    (283.82,108.71) -- (284.07,146.67) ;
\draw   (274.07,87.67) -- (294.07,87.67) -- (294.07,107.67) -- (274.07,107.67) -- cycle ;
\draw    (283.92,47.87) -- (283.82,87.71) ;
\draw    (309.92,123.55) -- (330.41,123.96) ;
\draw    (309.92,123.55) -- (309.92,143.95) ;
\draw    (309.92,143.95) -- (339.92,144.35) ;
\draw    (330.41,123.96) -- (339.92,144.35) ;
\draw   (283.44,28.6) -- (266.89,47.52) -- (300,47.52) -- cycle ;
\draw   (217,186.02) .. controls (217,171.67) and (228.63,160.04) .. (242.98,160.04) .. controls (257.33,160.04) and (268.96,171.67) .. (268.96,186.02) .. controls (268.96,200.37) and (257.33,212) .. (242.98,212) .. controls (228.63,212) and (217,200.37) .. (217,186.02) -- cycle ;
\draw (283.67,99.75) node   [align=left] {$\displaystyle D$};
\draw (321.84,133.84) node   [align=left] {$\displaystyle f$};
\draw (367.74,224.84) node   [align=left] {$\displaystyle \omega $};
\draw (282.45,39.88) node   [align=left] {$\displaystyle \omega ^{2}$};
\draw (242.91,182.09) node   [align=left] {$\displaystyle \frac{1}{\sqrt{|S} |}$};
\end{tikzpicture}
\]
The key operator $D$ can be decomposed in the following manner:

\[
\tikzset{every picture/.style={line width=0.75pt}} 
\begin{tikzpicture}[x=0.75pt,y=0.75pt,yscale=-1,xscale=1]
\draw    (195.82,145.95) -- (196.07,183.91) ;
\draw   (186.07,124.91) -- (206.07,124.91) -- (206.07,144.91) -- (186.07,144.91) -- cycle ;
\draw    (195.92,85.95) -- (195.82,124.95) ;
\draw    (281,90) -- (280.92,182.95) ;
\draw    (368.76,89.04) -- (368.76,125.04) ;
\draw [shift={(368.76,125.04)}, rotate = 90] [color={rgb, 255:red, 0; green, 0; blue, 0 }  ][fill={rgb, 255:red, 0; green, 0; blue, 0 }  ][line width=0.75]      (0, 0) circle [x radius= 3.35, y radius= 3.35]   ;
\draw    (368.76,143.04) -- (368.76,181.04) ;
\draw [shift={(368.76,143.04)}, rotate = 90] [color={rgb, 255:red, 0; green, 0; blue, 0 }  ][fill={rgb, 255:red, 0; green, 0; blue, 0 }  ][line width=0.75]      (0, 0) circle [x radius= 3.35, y radius= 3.35]   ;
\draw   (321,135.43) .. controls (321,123.59) and (330.59,114) .. (342.43,114) .. controls (354.26,114) and (363.85,123.59) .. (363.85,135.43) .. controls (363.85,147.26) and (354.26,156.85) .. (342.43,156.85) .. controls (330.59,156.85) and (321,147.26) .. (321,135.43) -- cycle ;
\draw (195.67,137) node   [align=left] {$\displaystyle D$};
\draw (239,128) node [anchor=north west][inner sep=0.75pt]   [align=left] {=};
\draw (343.43,130.43) node   [align=left] {$\displaystyle \frac{2}{|S|}$};
\draw (364,68) node [anchor=north west][inner sep=0.75pt]   [align=left] {$\displaystyle S$};
\draw (364,191) node [anchor=north west][inner sep=0.75pt]   [align=left] {$\displaystyle S$};
\draw (192,63) node [anchor=north west][inner sep=0.75pt]   [align=left] {$\displaystyle S$};
\draw (191,192) node [anchor=north west][inner sep=0.75pt]   [align=left] {$\displaystyle S$};
\draw (276,66) node [anchor=north west][inner sep=0.75pt]   [align=left] {$\displaystyle S$};
\draw (277,195) node [anchor=north west][inner sep=0.75pt]   [align=left] {$\displaystyle S$};
\draw (306,128) node [anchor=north west][inner sep=0.75pt]   [align=left] {$\displaystyle -$};
\end{tikzpicture}
\]
Similarly, this expression can be rewritten as follows:

\[
\tikzset{every picture/.style={line width=0.75pt}} 
\begin{tikzpicture}[x=0.75pt,y=0.75pt,yscale=-1,xscale=1]
\draw    (314,172) -- (314,215) ;
\draw [shift={(314,215)}, rotate = 90] [color={rgb, 255:red, 0; green, 0; blue, 0 }  ][fill={rgb, 255:red, 0; green, 0; blue, 0 }  ][line width=0.75]      (0, 0) circle [x radius= 3.35, y radius= 3.35]   ;
\draw    (395,71) -- (395,154) ;
\draw    (294,144) .. controls (295.76,154.04) and (300.76,170.04) .. (314,172) ;
\draw [shift={(314,172)}, rotate = 8.43] [color={rgb, 255:red, 0; green, 0; blue, 0 }  ][fill={rgb, 255:red, 0; green, 0; blue, 0 }  ][line width=0.75]      (0, 0) circle [x radius= 3.35, y radius= 3.35]   ;
\draw    (314,172) .. controls (325.76,168.04) and (330.76,160.04) .. (334,144) ;
\draw   (394.44,173) -- (377.89,154.08) -- (411,154.08) -- cycle ;
\draw    (293.82,107.95) -- (294.07,145.91) ;
\draw   (284.07,86.91) -- (304.07,86.91) -- (304.07,106.91) -- (284.07,106.91) -- cycle ;
\draw    (293.92,47.12) -- (293.82,86.95) ;
\draw    (319.92,122.79) -- (340.41,123.21) ;
\draw    (319.92,122.79) -- (319.92,143.19) ;
\draw    (319.92,143.19) -- (349.92,143.59) ;
\draw    (340.41,123.21) -- (349.92,143.59) ;
\draw   (293.44,27.84) -- (276.89,46.77) -- (310,46.77) -- cycle ;
\draw    (333.07,94) -- (332.82,122.04) ;
\draw   (331.44,75.84) -- (314.89,94.77) -- (348,94.77) -- cycle ;
\draw   (231,179.02) .. controls (231,164.67) and (242.63,153.04) .. (256.98,153.04) .. controls (271.33,153.04) and (282.96,164.67) .. (282.96,179.02) .. controls (282.96,193.37) and (271.33,205) .. (256.98,205) .. controls (242.63,205) and (231,193.37) .. (231,179.02) -- cycle ;
\draw   (330,176) .. controls (330,162.19) and (341.19,151) .. (355,151) .. controls (368.81,151) and (380,162.19) .. (380,176) .. controls (380,189.81) and (368.81,201) .. (355,201) .. controls (341.19,201) and (330,189.81) .. (330,176) -- cycle ;
\draw (293.67,99) node   [align=left] {$\displaystyle D$};
\draw (333.84,135.09) node   [align=left] {$\displaystyle f$};
\draw (293.37,40.08) node   [align=left] {$\displaystyle s$};
\draw (331.45,88.13) node   [align=left] {$\displaystyle \omega ^{2}$};
\draw (396.07,160.08) node   [align=left] {$\displaystyle \omega $};
\draw (256.91,175.09) node   [align=left] {$\displaystyle \frac{1}{\sqrt{|S} |}$};
\draw (356.44,175.09) node   [align=left] {$\displaystyle \frac{1}{\sqrt{3}}$};
\end{tikzpicture}
\]

Disregarding both the second subsystem and the scalar coefficients, we obtain:
\[
\tikzset{every picture/.style={line width=0.75pt}} 
\begin{tikzpicture}[x=0.65pt,y=0.65pt,yscale=-1,xscale=1]
\draw    (125.07,113) -- (125.07,156) ;
\draw [shift={(125.07,156)}, rotate = 90] [color={rgb, 255:red, 0; green, 0; blue, 0 }  ][fill={rgb, 255:red, 0; green, 0; blue, 0 }  ][line width=0.75]      (0, 0) circle [x radius= 3.35, y radius= 3.35]   ;
\draw    (105.07,85) .. controls (106.82,95.04) and (111.82,111.04) .. (125.07,113) ;
\draw [shift={(125.07,113)}, rotate = 8.43] [color={rgb, 255:red, 0; green, 0; blue, 0 }  ][fill={rgb, 255:red, 0; green, 0; blue, 0 }  ][line width=0.75]      (0, 0) circle [x radius= 3.35, y radius= 3.35]   ;
\draw    (125.07,113) .. controls (136.82,109.04) and (141.82,101.04) .. (145.07,85) ;
\draw   (97.07,65) -- (117.07,65) -- (117.07,85) -- (97.07,85) -- cycle ;
\draw    (106.92,36.15) -- (106.82,65.04) ;
\draw   (103.44,15.84) -- (86.89,34.77) -- (120,34.77) -- cycle ;
\draw    (131.92,64.79) -- (152.41,65.21) ;
\draw    (131.92,64.79) -- (131.92,85.19) ;
\draw    (131.92,85.19) -- (161.92,85.59) ;
\draw    (152.41,65.21) -- (161.92,85.59) ;
\draw    (145.07,36) -- (144.82,64.04) ;
\draw   (143.44,17.84) -- (126.89,36.77) -- (160,36.77) -- cycle ;
\draw    (261.07,119) -- (261.07,162) ;
\draw [shift={(261.07,162)}, rotate = 90] [color={rgb, 255:red, 0; green, 0; blue, 0 }  ][fill={rgb, 255:red, 0; green, 0; blue, 0 }  ][line width=0.75]      (0, 0) circle [x radius= 3.35, y radius= 3.35]   ;
\draw    (241.07,91) .. controls (242.82,101.04) and (247.82,117.04) .. (261.07,119) ;
\draw [shift={(261.07,119)}, rotate = 8.43] [color={rgb, 255:red, 0; green, 0; blue, 0 }  ][fill={rgb, 255:red, 0; green, 0; blue, 0 }  ][line width=0.75]      (0, 0) circle [x radius= 3.35, y radius= 3.35]   ;
\draw    (261.07,119) .. controls (272.82,115.04) and (277.82,107.04) .. (281.07,91) ;
\draw   (239.44,21.84) -- (222.89,40.77) -- (256,40.77) -- cycle ;
\draw    (267.92,70.79) -- (288.41,71.21) ;
\draw    (267.92,70.79) -- (267.92,91.19) ;
\draw    (267.92,91.19) -- (297.92,91.59) ;
\draw    (288.41,71.21) -- (297.92,91.59) ;
\draw    (281.07,42) -- (280.82,70.04) ;
\draw   (279.44,23.84) -- (262.89,42.77) -- (296,42.77) -- cycle ;
\draw    (240.92,41.06) -- (241.07,91) ;
\draw    (402.07,117) -- (402.07,160) ;
\draw [shift={(402.07,160)}, rotate = 90] [color={rgb, 255:red, 0; green, 0; blue, 0 }  ][fill={rgb, 255:red, 0; green, 0; blue, 0 }  ][line width=0.75]      (0, 0) circle [x radius= 3.35, y radius= 3.35]   ;
\draw    (382.07,89) .. controls (383.82,99.04) and (388.82,115.04) .. (402.07,117) ;
\draw [shift={(402.07,117)}, rotate = 8.43] [color={rgb, 255:red, 0; green, 0; blue, 0 }  ][fill={rgb, 255:red, 0; green, 0; blue, 0 }  ][line width=0.75]      (0, 0) circle [x radius= 3.35, y radius= 3.35]   ;
\draw [shift={(382.07,89)}, rotate = 80.07] [color={rgb, 255:red, 0; green, 0; blue, 0 }  ][fill={rgb, 255:red, 0; green, 0; blue, 0 }  ][line width=0.75]      (0, 0) circle [x radius= 3.35, y radius= 3.35]   ;
\draw    (402.07,117) .. controls (413.82,113.04) and (418.82,105.04) .. (422.07,89) ;
\draw    (383.92,40.15) -- (383.82,69.04) ;
\draw [shift={(383.82,69.04)}, rotate = 90.19] [color={rgb, 255:red, 0; green, 0; blue, 0 }  ][fill={rgb, 255:red, 0; green, 0; blue, 0 }  ][line width=0.75]      (0, 0) circle [x radius= 3.35, y radius= 3.35]   ;
\draw   (380.44,19.84) -- (363.89,38.77) -- (397,38.77) -- cycle ;
\draw    (408.92,68.79) -- (429.41,69.21) ;
\draw    (408.92,68.79) -- (408.92,89.19) ;
\draw    (408.92,89.19) -- (438.92,89.59) ;
\draw    (429.41,69.21) -- (438.92,89.59) ;
\draw    (422.07,40) -- (421.82,68.04) ;
\draw   (420.44,21.84) -- (403.89,40.77) -- (437,40.77) -- cycle ;
\draw    (226.92,215.79) -- (247.41,216.21) ;
\draw    (226.92,215.79) -- (226.92,236.19) ;
\draw    (226.92,236.19) -- (256.92,236.59) ;
\draw    (247.41,216.21) -- (256.92,236.59) ;
\draw    (240.07,187) -- (239.82,215.04) ;
\draw   (238.44,168.84) -- (221.89,187.77) -- (255,187.77) -- cycle ;
\draw    (242.07,238.09) -- (241.82,266.13) ;
\draw   (241.44,283) -- (224.89,264.08) -- (258,264.08) -- cycle ;
\draw    (336.92,220.95) -- (357.41,221.36) ;
\draw    (336.92,220.95) -- (336.92,241.35) ;
\draw    (336.92,241.35) -- (366.92,241.75) ;
\draw    (357.41,221.36) -- (366.92,241.75) ;
\draw    (350.07,192.16) -- (349.82,220.19) ;
\draw   (348.44,174) -- (331.89,192.92) -- (365,192.92) -- cycle ;
\draw    (352.07,243.24) -- (351.82,271.28) ;
\draw [shift={(351.82,271.28)}, rotate = 90.5] [color={rgb, 255:red, 0; green, 0; blue, 0 }  ][fill={rgb, 255:red, 0; green, 0; blue, 0 }  ][line width=0.75]      (0, 0) circle [x radius= 3.35, y radius= 3.35]   ;
\draw   (334,84.43) .. controls (334,72.59) and (343.59,63) .. (355.43,63) .. controls (367.26,63) and (376.85,72.59) .. (376.85,84.43) .. controls (376.85,96.26) and (367.26,105.85) .. (355.43,105.85) .. controls (343.59,105.85) and (334,96.26) .. (334,84.43) -- cycle ;
\draw   (289,227.43) .. controls (289,215.59) and (298.59,206) .. (310.43,206) .. controls (322.26,206) and (331.85,215.59) .. (331.85,227.43) .. controls (331.85,239.26) and (322.26,248.85) .. (310.43,248.85) .. controls (298.59,248.85) and (289,239.26) .. (289,227.43) -- cycle ;
\draw (106.67,77.08) node   [align=left] {$\displaystyle D$};
\draw (103.37,28.08) node   [align=left] {$\displaystyle s$};
\draw (142.84,75.09) node   [align=left] {$\displaystyle f$};
\draw (239.37,34.08) node   [align=left] {$\displaystyle s$};
\draw (281.84,81.09) node   [align=left] {$\displaystyle f$};
\draw (316,75) node [anchor=north west][inner sep=0.75pt]   [align=left] {$\displaystyle -$};
\draw (380.37,32.08) node   [align=left] {$\displaystyle s$};
\draw (421.84,78.09) node   [align=left] {$\displaystyle f$};
\draw (180,78) node [anchor=north west][inner sep=0.75pt]   [align=left] {$\displaystyle =$};
\draw (180,209) node [anchor=north west][inner sep=0.75pt]   [align=left] {$\displaystyle =$};
\draw (239.84,226.09) node   [align=left] {$\displaystyle f$};
\draw (241.37,271.93) node   [align=left] {$\displaystyle s$};
\draw (350.84,230.24) node   [align=left] {$\displaystyle f$};
\draw (268,218) node [anchor=north west][inner sep=0.75pt]   [align=left] {$\displaystyle -$};
\draw (420.45,33.13) node   [align=left] {$\displaystyle \omega ^{2}$};
\draw (145.45,29.13) node   [align=left] {$\displaystyle \omega ^{2}$};
\draw (279.44,33.31) node   [align=left] {$\displaystyle \omega ^{2}$};
\draw (236.45,180.13) node   [align=left] {$\displaystyle \omega ^{2}$};
\draw (348.44,183.46) node   [align=left] {$\displaystyle \omega ^{2}$};
\draw (356.43,79.43) node   [align=left] {$\displaystyle \frac{2}{|S|}$};
\draw (311.43,222.43) node   [align=left] {$\displaystyle \frac{2}{|S|}$};
\end{tikzpicture}
\]

For a balanced element $s$, we have:

\[
\tikzset{every picture/.style={line width=0.75pt}} 
\begin{tikzpicture}[x=0.75pt,y=0.75pt,yscale=-1,xscale=1]
\draw    (224.92,152.95) -- (245.41,153.36) ;
\draw    (224.92,152.95) -- (224.92,173.35) ;
\draw    (224.92,173.35) -- (254.92,173.75) ;
\draw    (245.41,153.36) -- (254.92,173.75) ;
\draw    (238.07,124.16) -- (237.82,152.19) ;
\draw   (236.44,106) -- (219.89,124.92) -- (253,124.92) -- cycle ;
\draw    (240.07,175.24) -- (239.82,203.28) ;
\draw   (239.44,220.16) -- (222.89,201.23) -- (256,201.23) -- cycle ;
\draw    (334.92,158.11) -- (355.41,158.52) ;
\draw    (334.92,158.11) -- (334.92,178.51) ;
\draw    (334.92,178.51) -- (364.92,178.91) ;
\draw    (355.41,158.52) -- (364.92,178.91) ;
\draw    (348.07,129.31) -- (347.82,157.35) ;
\draw   (346.44,111.16) -- (329.89,130.08) -- (363,130.08) -- cycle ;
\draw    (350.07,180.4) -- (349.82,208.44) ;
\draw [shift={(349.82,208.44)}, rotate = 90.5] [color={rgb, 255:red, 0; green, 0; blue, 0 }  ][fill={rgb, 255:red, 0; green, 0; blue, 0 }  ][line width=0.75]      (0, 0) circle [x radius= 3.35, y radius= 3.35]   ;
\draw   (291,166.43) .. controls (291,154.59) and (300.59,145) .. (312.43,145) .. controls (324.26,145) and (333.85,154.59) .. (333.85,166.43) .. controls (333.85,178.26) and (324.26,187.85) .. (312.43,187.85) .. controls (300.59,187.85) and (291,178.26) .. (291,166.43) -- cycle ;
\draw (237.84,162.24) node   [align=left] {$\displaystyle f$};
\draw (239.37,209.08) node   [align=left] {$\displaystyle s$};
\draw (348.84,167.4) node   [align=left] {$\displaystyle f$};
\draw (271,156) node [anchor=north west][inner sep=0.75pt]   [align=left] {$\displaystyle =$};
\draw (235.45,118.13) node   [align=left] {$\displaystyle \omega ^{2}$};
\draw (344.45,123.13) node   [align=left] {$\displaystyle \omega ^{2}$};
\draw (313.43,161.43) node   [align=left] {$\displaystyle \frac{2}{|S|}$};
\end{tikzpicture}
\]
This implies that following the measurement, such a balanced element $s$ cannot be observed.
For the case where $G = \mathbb{Z}_3$, suppose that the image $f(S)$ contains $a$ 0's, $b$ 1's,  $c$ 2's and under the condition that $\omega^2(f(s))=1$, we obtain
\[
\frac{2(a+b\omega^2+c\omega)}{a+b+c} = 1.
\]
Consequently, it is straightforward to see that if $a : b : c = 4 : 1 : 1$, the number of balanced elements is given by $a$.

In what follows, we extend this construction to a more general formulation.
\begin{theorem}
If n is a prime number, then
generalized single-shot Grover algorithm is guaranteed to return an unbalanced element of S
if and only if $a_0=(n+1)a_1=\cdots=(n+1)a_{n-1}$ where $a_i $ is the number of $s$ satisfying $f(s)=i$.
\begin{proof}

For a general group $G = \mathbb{Z}_n$, following a simplification similar to that of the $\mathbb{Z}_3$ case, we replace the effect $\omega^2$ with an irreducible representation $\rho$ satisfying $\rho(i) = q^i$, where $q$ is a primitive $n$-th root of unity. Without loss of generality, assume that $f(S)$ contains $a_0$ 0's, $a_1$ 1's, $\ldots$ , $a_{n-1}$ $n-1$'s and $\rho(f(s))=1$, according
to the balanced formula, we have
\[
\frac{2\sum_{i=0}^{n-1} a_i q^i}{\sum_{i=0}^{n-1} a_i} = 1\\
\Leftrightarrow \sum_{i=0}^{n-1} a_i(2q^i-1)=0.
\]
Note that $\sum_{i=0}^{n-1} a_i(2x^i-1)$ is a  polynomial with $q$ as a root, it must be a product of cyclotomic polynomials.
When $n$ is a prime number, the only cyclotomic polynomial $\Phi_{n}(x) = 1 + x + x^{2} + \cdots + x^{n-1} $. So we can conclude that $a_0=(n+1)a_1=\cdots=(n+1)a_{n-1}$.
Conversely, if $a_0=(n+1)a_1=\cdots=(n+1)a_{n-1}$, this generalized single-shot Grover algorithm is guaranteed to return an unbalanced element of $S$.
\end{proof}
\end{theorem}

 Therefore, if we mark elements according to the proportion described above, this generalized single-shot Grover algorithm is guaranteed to return an unbalanced element of $S$.
However, when $n$ is not prime, $a_0=(n+1)a_1=\cdots=(n+1)a_{n-1}$ is not a suffcient condition.

\section{Controlled-NOT gates and quantum entanglement}
In the preceding three sections on quantum algorithms, our discussion of quantum advantage has consistently highlighted the crucial role of quantum entanglement. In this section, we further explore the entanglement properties of qubits from the perspective of tensor diagram semantics.

\subsection{Controlled-NOT gates}
Coecke and Pavlovic introduced Frobenius structures into quantum mechanics in 2007, involving three stages: preparation, evolution, and measurement\cite{CP07}. In 2011, Duncan and Coecke further developed the graphical language, giving representations of quantum gates such as the CNOT gate from the ZX calculus perspective\cite{CD11}.
\begin{center}
\includegraphics[width=9cm]{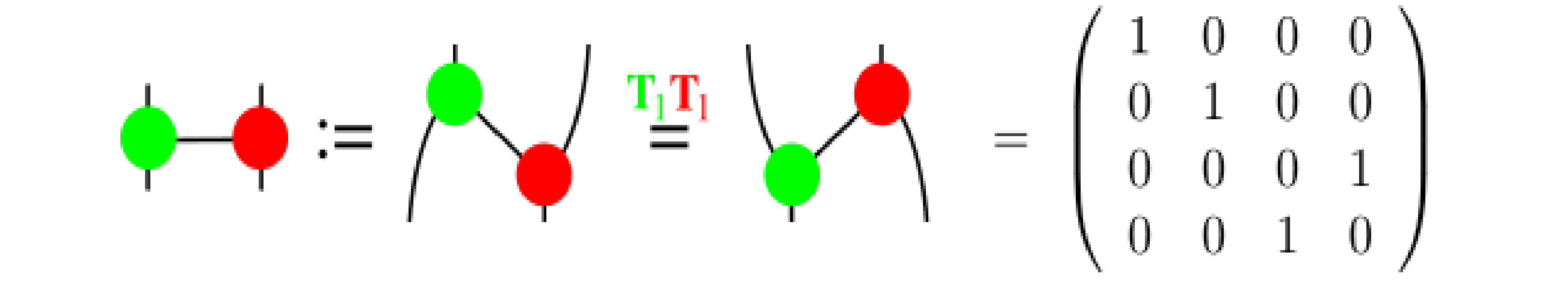}
\end{center}

The tensor diagram semantics of the Bernstein-Vazirani oracle is given as follows:
\begin{center}

\begin{tikzpicture}[thick,scale=0.6]

    \draw (0,0) rectangle (2,2);
    \node at (1,1) {\large $U_f$};
    \draw (-0.8, 1.7) -- (0,1.7) node[pos=0, left] {\small $x$};
    \draw (2, 1.7) -- (2.8,1.7) node[pos=1, right] {\small $x$};
    \draw (-0.8, 0.3) -- (0,0.3) node[pos=0, left] {\small $y$};
    \draw (2, 0.3) -- (2.8,0.3) node[pos=1, right] {\small $y \oplus f(x)$};

    \draw[<->] (5.5,1) -- (6.5,1);

    \begin{scope}[shift={(8,-1.25)}]
        \draw (0,3.5) -- (0,2) to [out=-90,in=180] (0.8,1.2);
        \draw (0.8,1.2) -- (0.8,0.5); 
        \filldraw[black] (0.8,1.2) circle (2.5pt);
        \draw (0.8,1.2) to [out=0,in=-90] (1.6,2) ;
        \draw (1.2,2.7) -- (1.2,2) -- (2.1,2) -- (1.9,2.7) -- cycle;
        \node at (1.6,2.3) {$f$};
        \draw (1.6,2.7) to [out=90,in=180] (2.3,3.5);
        \draw (2.4,3.6) -- (2.4,4);
        \node[draw, circle, fill=white, inner sep=1.5pt] at (2.4,3.5) {};
        \draw (2.5,3.5) to [out=0,in=90] (3.2,2.7) -- (3.2,0.8);
    \end{scope}

\end{tikzpicture}
\end{center}
From the perspective of tensor diagram semantics, we present the following theorem. Although this result has been proven previously, we provide an alternative proof that avoids purely matrix-based calculations.
\begin{theorem}\label{5.1}
In $\mathbf{FHilb}$ and $A=(\mathbb{C}^2)^{\otimes n}$, a pair of complementary symmetric dagger Frobenius structures satisfies:
\begin{center}
\begin{tikzpicture}[thick, scale=0.9]

    \begin{scope}[shift={(0,0)}]
        \draw (0,3) -- (0,2) to [out=-90,in=180] (0.8,1.2);
        \draw (0.7,1.2) -- (0.7,0.5); 
        \filldraw[black] (0.7,1.2) circle (2.5pt);

        \draw (0.7,1.2) to [out=0,in=-90] (1.3,1.8) to [out=90,in=180] (2,2.4);
        \draw (2,2.4) -- (2,3);
        \node[draw, circle, fill=white, inner sep=2pt] at (2,2.4) {};
        \draw (2.1,2.4) to [out=0,in=90] (2.7,1.5) -- (2.7,0.5);
    \end{scope}

    \node at (3.8, 1.75) {\large $=$};

    \begin{scope}[shift={(5,0)}]
        \draw (0,0.5) -- (0,1.5) to [out=90,in=180] (0.7,2.4);
        \draw (0.7,2.4) -- (0.7,3); 
        \filldraw[black] (0.7,2.4) circle (2.5pt); 
        \draw (0.8,2.4) to [out=0,in=90] (1.3,1.8) to [out=-90,in=180] (2,1.2);
        \draw (2,1.2) -- (2,0.6);
        \node[draw, circle, fill=white, inner sep=2pt] at (2,1.2) {};
        \draw (2.1,1.2) to [out=0,in=-90] (2.7,2)-- (2.7,3);
    \end{scope}

\end{tikzpicture}
\end{center}
\end{theorem}

\subsection{Quantum entanglement}
Quantum entanglement is a fundamental feature of quantum mechanics and is closely tied to the application of controlled-NOT (CNOT) gates. Canonical examples of entangled systems include the Bell states and the Greenberger-Horne-Zeilinger (GHZ) state. Their definitions are given as follows:
$$
\vert \mathrm{Bell} \rangle = \frac{\vert 00 \rangle + \vert 11 \rangle}{\sqrt{2}},\quad \qquad \vert \mathrm{GHZ} \rangle = \frac{\vert 000 \rangle + \vert 111 \rangle}{\sqrt{2}}.
$$
The quantum circuits for preparing these states are given as follows:
\begin{center}

\end{center}

The tensor diagram semantics of the Bell state and GHZ state are represented as follows\cite{CK17}:
\begin{center}
%
\end{center}

We now consider the specific case of a single qubit. In the category $\mathbf{FHilb}$, let $A = \mathbb{C}^2$, equipped with the Frobenius structures $(A, \mu_{B}, \eta_{B})$ and $(A, \mu_{W}, \eta_{W})$ associated with the computational basis and the $\sqrt{2}\mathrm{X}$ basis, respectively. A straightforward computation then yields:
\begin{center}
\begin{equation}\label{qubit}
%
\end{equation}
\end{center}
Using this computation, we now prove Theorem \ref{5.1}.
\begin{proof}[Proof of Theorem \ref{5.1}]
\[
%
\]
\noindent where the first and third equalities both use the non-commutative spider theorem \ref{thm2.2}, and the second equality uses \ref{qubit} .
\end{proof}

Next, we consider another fundamental three-qubit entangled state: the W state. The quantum circuit for preparing this state is given as follows\cite{Cruz19}:

\[
\tikzset{every picture/.style={line width=0.75pt}} 
%
\]
where $\theta = 2\arccos(\frac{1}{\sqrt{3}})$.

The corresponding ZX-calculus representation is given as follows. Throughout the subsequent simplifications, we omit global phases and scalar factors.

\[
\tikzset{every picture/.style={line width=0.75pt}} 
%
\]
In the aforementioned diagrammatic simplification, we employ basic ZX-calculus rules, phase gadgets, local complementation and pivot rule \cite{Wet20}, where the dashed lines represent Hadamard edges. Ultimately, the resulting diagram yields an expression distinct from the standard formulation of W state introduced by Coecke  \cite{CE10}.

\section{Conclusion}

This paper investigates foundational quantum protocols through categorical tensor-graph semantics in \textbf{FHilb}, demonstrating that algorithmic functionalities can be explicitly distilled into topological skeletons rather than obscured by high-dimensional matrices. By topologically reinterpreting the Bernstein-Vazirani and Simon algorithms, we elucidated the intuitive diagrammatic mechanisms of phase kickback and information retrieval. Extending beyond the standard binary qubit paradigm, we applied our diagrammatic toolkit to multi-valued quantum systems, successfully formalizing the qutrit-adapted generalized Deutsch--Jozsa. Notably, by analyzing the topological representation of the generalized Grover algorithm, we derived rigorous algebraic conditions for its deterministic success, seamlessly linking diagrammatic algorithmic design with group representation theory. Furthermore, we explored the categorical genesis of quantum entanglement by analyzing controlled-NOT gates using complementary Frobenius structures, and we investigated a diagrammatic decomposition of the W-state preparation protocol.

Specifically, our W-state reduction retained irreducible non-Clifford phase gates, diverging from standard simplified forms. This phenomenon not only highlights the more intricate topological entanglement properties inherent in the W-state compared to GHZ or Bell states, but also indicates current limitations in pure graphical axiomatic systems when handling highly complex non-Clifford resources. Consequently, resolving these residual non-Clifford elements presents a clear trajectory for future research: enriching tensor category axioms with higher-order rewrite rules to enable automated, scalable circuit optimization in modern quantum compilers.

\section*{\textbf{Conflict of interest}}
	
	The authors declare that they have no conflict of interest.

\section*{\textbf{Acknowledgments}}
We are grateful to Dr. Fen Zuo (Shanghai Micro-Era Digital Technology Co., Ltd.) for his constructive comments and valuable suggestions on the first version of the manuscript. This work is supported by Fundamental and Interdisciplinary Disciplines Breakthrough Plan of the Ministry of Education of China (JYB2025XDXM112),
		supported by the NNSF of China (Grant No. 12171155), and in part by the Science and Technology Commission of Shanghai Municipality (Grant No. 22DZ2229014).
\bibliographystyle{plain}

\onecolumn\newpage
\appendix
\section{Categorical tensor-graphs}

Below, we review the foundational concepts of categorical tensor-graphs; a more detailed exposition can be found in  \cite{HV19}.
\begin{definition}
In the monoidal categories $\mathbf{Hilb}$ and $\mathbf{FHilb}$:

$\bullet$ The tensor product $\otimes: \mathbf{Hilb} \times \mathbf{Hilb}\rightarrow\mathbf{Hilb}$ is the tensor product of Hilbert spaces;

$\bullet$ The unit object $I$ is the one-dimensional Hilbert space $\mathbb{C}$;

$\bullet$ The associator $(H \otimes J) \otimes K \stackrel{\alpha_{H, J, K}}{\longrightarrow} H \otimes(J \otimes K)$ is the unique linear map satisfying $(a \otimes b) \otimes c \mapsto a \otimes(b \otimes c)$ for all $\forall a \in H,  b \in J$ and $c \in K$;

$\bullet$ The left unitor $\mathbb{C} \otimes H \stackrel{\lambda_{H}}{\longrightarrow} H$ is the unique linear map satisfying $1 \otimes a \mapsto a$ for all $\forall a \in H$;

$\bullet$ The right unitor $H \otimes \mathbb{C} \stackrel{\rho_{H}}{\longrightarrow} H$ is the unique linear map satisfying $a \otimes 1 \mapsto a$ for all $\forall a \in H$.\\
$\mathbf{FHilb}$  denotes the restriction of $\mathbf{Hilb}$ to finite-dimensional spaces.
\end{definition}

For morphisms $A \stackrel{f}{\rightarrow} B$ and $C \stackrel{g}{\rightarrow} D$, we draw their tensor product $A \otimes C \stackrel{f \otimes g}{\longrightarrow} B \otimes D$ as:
\begin{center}
\includegraphics[width=13cm]{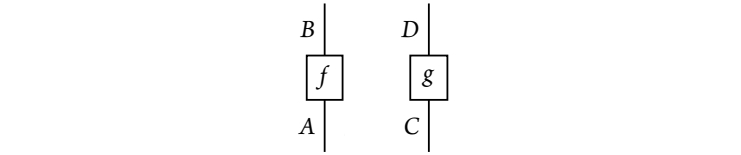}
\end{center}
where inputs are drawn at the bottom and outputs at the top; in this sense, ``time" flows from bottom to top.

Below we introduce two special kinds of morphisms: states and effects, which we will frequently use in the following to simulate the preparation and measurement processes of quantum algorithms.
\begin{definition}[state]
In a monoidal category, a state of object $A$ is a morphism is a morphismx $I\rightarrow A$.
\end{definition}
\begin{definition}[effect]
In a monoidal category, an effect or co-state of object $A$ is a morphism $A \rightarrow I$.
\end{definition}

Their graphical calculus representations are as follows
\begin{center}
\begin{tikzpicture}[thick, scale=1.2]

    \begin{scope}[shift={(0,0)}]
        \draw (0, 1) -- (0, 0.4);

        \draw (0, -0.1) -- (-0.4, 0.4) -- (0.4, 0.4) -- cycle;

        \node at (0, 0.22) {\small $a$};
    \end{scope}

    \begin{scope}[shift={(2.5,0)}]
        \draw (0, 0.9) -- (-0.4, 0.4) -- (0.4, 0.4) -- cycle;

        \node at (0, 0.58) {\small $a$};

        \draw (0, 0.4) -- (0, -0.2);
    \end{scope}

\end{tikzpicture}
\end{center}
In graphical calculus, we draw the unit object $I$ as an empty diagram, so the domain and codomain of the above two morphisms are empty diagrams. To emphasize this, we use triangles instead of the original squares.

$\mathrm{Dagger}$ is also an important concept in category theory. After introducing this property, the drawing of morphisms will be further modified.
\begin{definition}[$\mathrm{dagger}$]
A $\mathrm{dagger}$ on a category $\mathbf{C}$ is an involutive contravariant functor $\dagger: \mathbf{C} \rightarrow \mathbf{C}$ that is identity on objects. A $\mathrm{dagger}$ category is a category equipped with a $\mathrm{dagger}$ .
\end{definition}

In graphical calculus, we represent taking the dagger  by reflecting the diagram across a horizontal axis. To help distinguish these morphisms, we will draw the original boxes in an asymmetric way.
\begin{center}
\includegraphics[width=12cm]{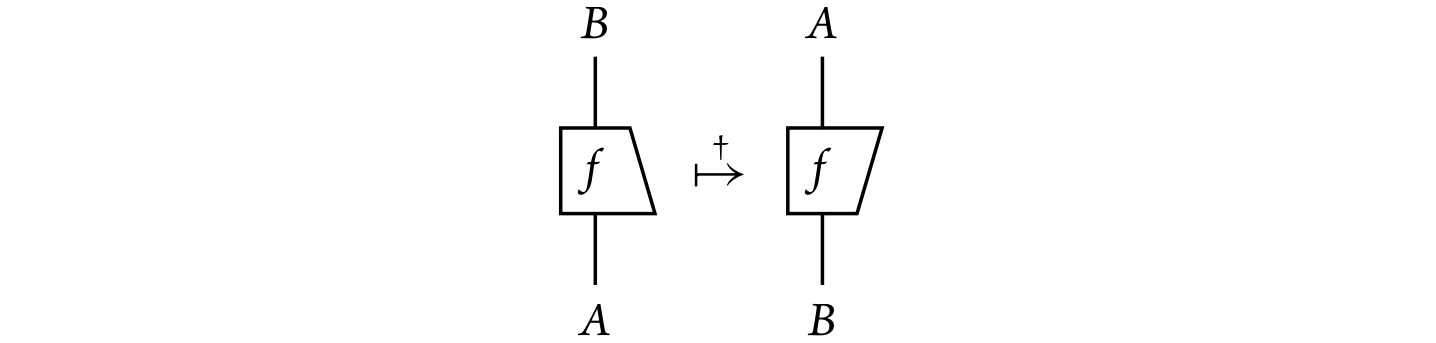}
\end{center}

The concept of copyable states introduced below has a crucial connection with the entanglement.
\begin{definition}[Copyable state]
In a braided monoidal category, given an object $A$ and a morphism $A \xrightarrow{d} A \otimes A$, a state $I \xrightarrow{a} A$ is called copyable when the following condition holds:
\begin{center}
\includegraphics[width=11cm]{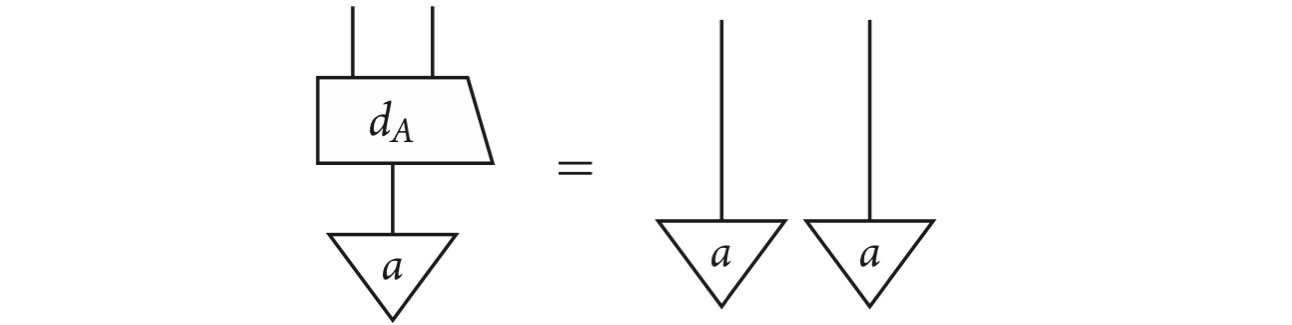}
\end{center}
\end{definition}

Frobenius structures are the foundation for simulating quantum computation, and the comonoid and monoid structures are precisely the building blocks of Frobenius structures. Moreover, due to the important roles of the four morphisms multiplication, unit, comultiplication, and counit, we use four new tensor diagrams to represent them:
\begin{center}
\begin{tikzpicture}[thick, scale=1]
        \node at (-1.5,0) {$\Delta_W :$};
    \begin{scope}[shift={(0,0)}]
        \draw (-0.7,0.6) to [out=-90,in=180] (0.1,0);
        \draw (0.1,0) to [out=0,in=-90] (0.7,0.6);
        \node[draw, circle, fill=white, inner sep=1.5pt] at (0,0) {};
        \draw (0,-0.1) -- (0,-0.5); 
    \end{scope}

        \node at (2,0) {$\varepsilon_W :$};

    \begin{scope}[shift={(3,0)}]
        \node[draw, circle, fill=white, inner sep=1.5pt] at (0,0) {};
        \draw (0,-0.1) -- (0,-0.5); 
    \end{scope}

        \node at (4.5,0) {$\mu_B :$};

    \begin{scope}[shift={(6,0)}]
        \draw (-0.7,-0.6) to [out=90,in=180] (0,0);
        \draw (0,0) to [out=0,in=90] (0.7,-0.6);
        \filldraw[black] (0,0) circle (2pt);
        \draw (0,0) -- (0,0.5); 
    \end{scope}

        \node at (8,0) {$\eta_B :$};

    \begin{scope}[shift={(9,0)}]
        \filldraw[black] (0,0) circle (2pt);
        \draw (0,0) -- (0,0.5); 
    \end{scope}
\end{tikzpicture}
\end{center}
\begin{definition}[Comonoid]
In a monoidal category, a comonoid is a triple $(A, \Delta, \varepsilon)$, where $A$ is an object and morphisms  $\Delta:A \to A \otimes A$ and $\varepsilon: A \to I$ satisfy the following equalities called coassociativity and counit laws
\begin{center}
\includegraphics[width=5.5cm]{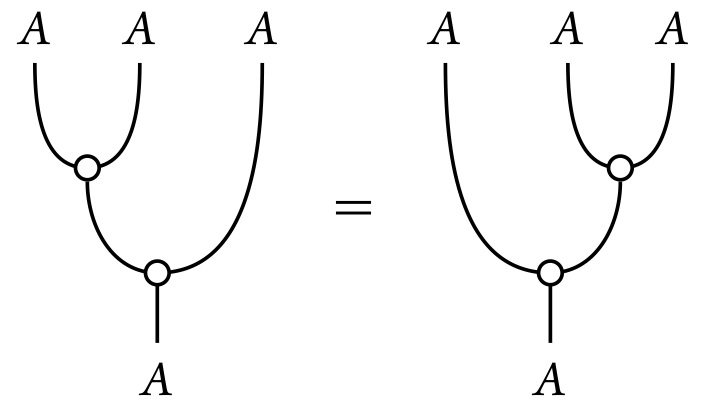}
\hspace{0.05\textwidth}
\includegraphics[width=5cm]{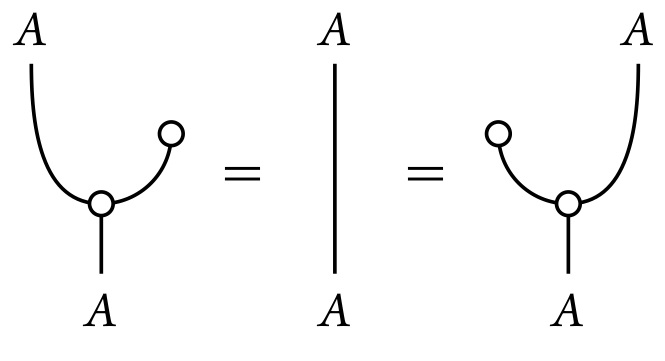}
\end{center}
If the monoidal category is a braided monoidal category, and the following holds
\begin{center}
\includegraphics[width=10cm]{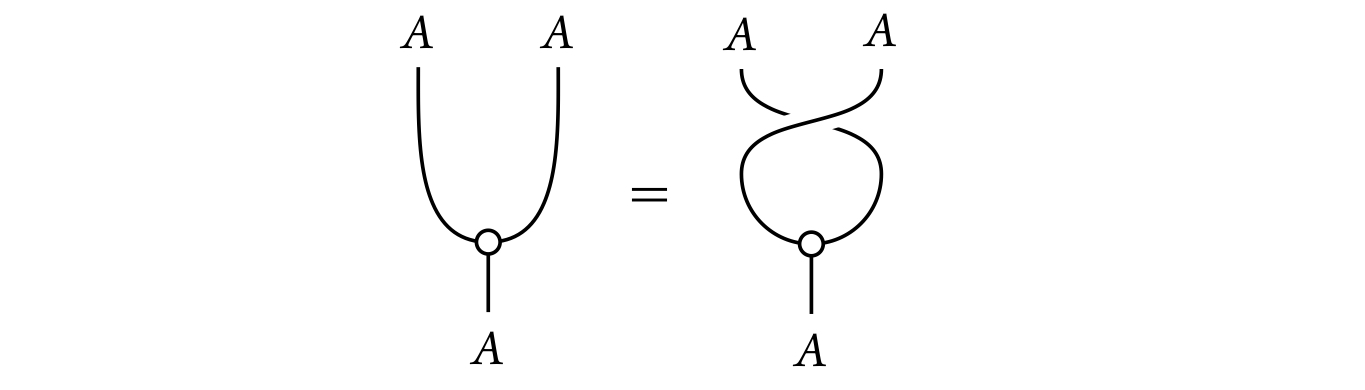}
\end{center}
then this comonoid is called cocommutative. The morphism $\Delta$ is called comultiplication, and the morphism $\varepsilon$ is called counit.
\end{definition}
\begin{definition}[Monoid]
In a monoidal category, a monoid is a triple $(A , \mu , \eta)$, where $A$ is an object, morphisms $\mu:A \otimes A \to A$ and state $\eta: I \to A$ satisfy the following equalities called associativity and unit laws:
\begin{center}
\includegraphics[width=6cm]{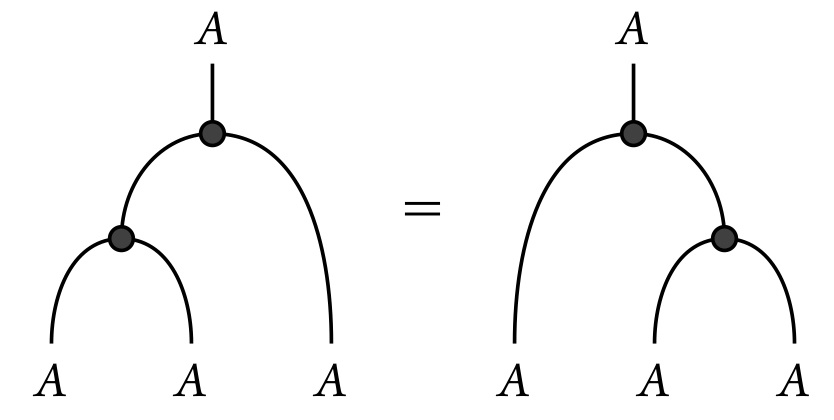}
\hspace{0.03\textwidth}
\includegraphics[width=5.5cm]{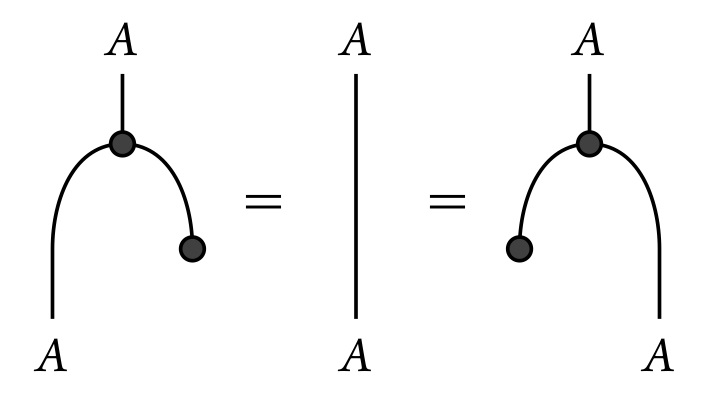}
\end{center}
The morphism $\mu$ is called multiplication, and the morphism $\eta$ is called unit.
\end{definition}
\begin{definition}[Frobenius structure in diagrams]
In a monoidal category, a Frobenius structure is a pair consisting of a comonoid  $(A,  \Delta_{W},  \varepsilon_{W})$ and a monoid $(A,  \Delta_{W},  \varepsilon_{W})$ satisfying the following equalities called Frobenius laws:
\begin{center}
\includegraphics[width=8cm]{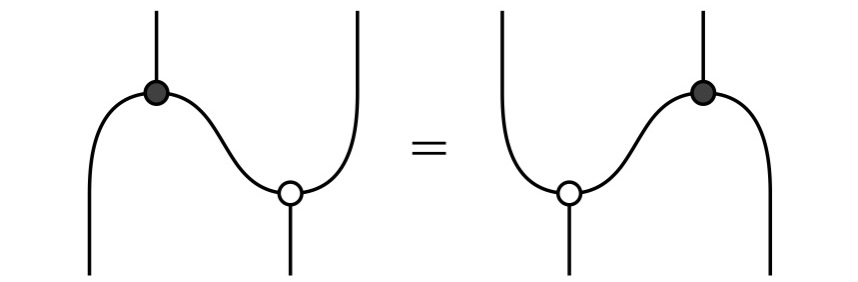}
\end{center}
\end{definition}
\begin{lemma}[Extended Frobenius laws]\label{lem2.1}
In a monoidal category, a Frobenius structure satisfies the following equalities:
\begin{center}
\includegraphics[width=9.5cm]{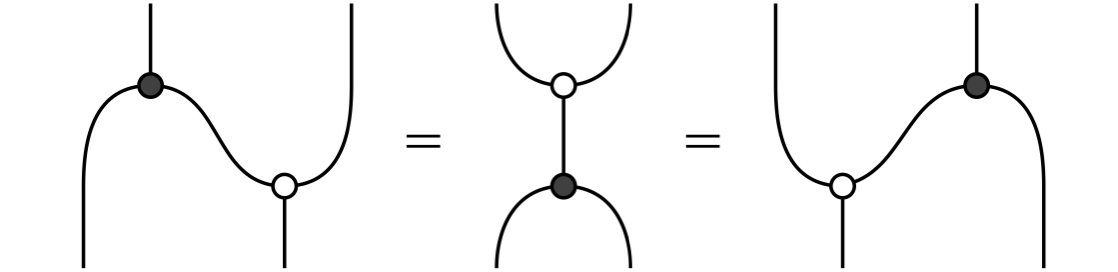}
\end{center}
\end{lemma}
\begin{theorem}[Non-commutative spider theorem]\label{thm2.2}
In a monoidal category, let $(A,  \mu_{B},  \eta_{B},$ $\Delta_{W},  \varepsilon_{W})$ be a Frobenius structure. For any connected morphism $A^{\otimes m} \rightarrow A^{\otimes n}$ built from finitely many $\mu_{B},  \eta_{B},  \Delta_{W},  \varepsilon_{W}$ and $\mathrm{id}$, connected by $\mu_{B},  \eta_{B},  \Delta_{W},  \varepsilon_{W}$, $\mathrm{id}$, $\circ$ and $\otimes$, it is equal to the following normal form:
\begin{center}
\includegraphics[width=7.5cm]{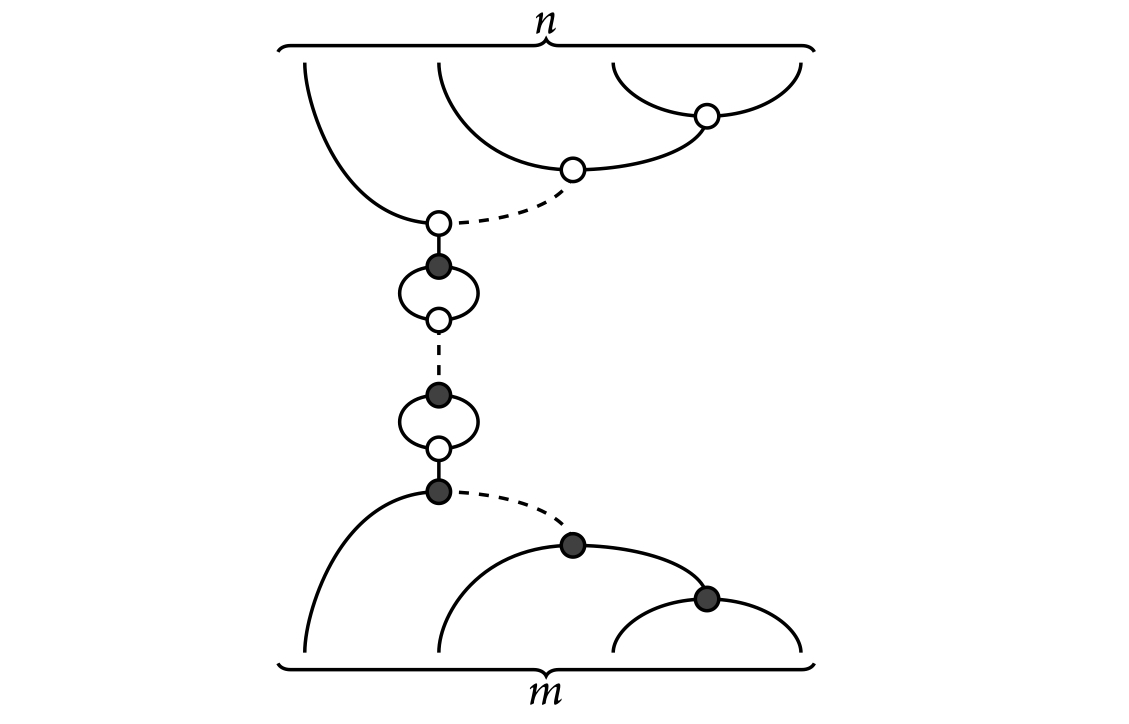}
\end{center}
\end{theorem}

For Frobenius structures, their various properties guarantee that different graphical transformations can be performed during topological deformations, allowing tensor diagrams to be deformed into desired forms and structures. Below we introduce several very important properties that appear frequently and play key roles in this paper.
\begin{definition}[Special]
In a monoidal category, a pair consisting of a monoid $(A,  \mu_{B},  \eta_{B})$ and a comonoid $(A,  \Delta_{W},  \varepsilon_{W})$ is called special when $\mu_{B}$ is a left inverse of $\Delta_{W}$:
\begin{center}
\includegraphics[width=7cm]{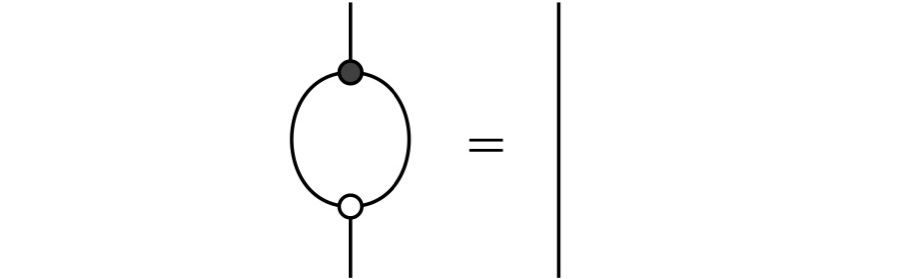}
\end{center}

Note that specialness and the Frobenius  law are the only two canonical interactions possible between a single multiplication and comultiplication.
\end{definition}
\begin{definition}[$\mathrm{dagger}$\ Frobenius structure]
In a monoidal dagger category, a Frobenius structure $(A,  \mu_{B},  \eta_{B},  \Delta_{W},  \varepsilon_{W})$ is called a dagger Frobenius structure when $\mu_{B}=(\Delta_{W})^{\dagger}$ and $\eta_{B}=(\varepsilon_{W})^{\dagger}$.
\end{definition}
\begin{definition}[Classical structure]
In a braided monoidal category, a classical structure is a special and commutative dagger Frobenius structure.
\end{definition}

Next we introduce a very important theorem, which reveals the isomorphism between Frobenius structures in $\mathbf{FHilb}$ and bases of finite-dimensional Hilbert spaces. Through this theorem, in the proofs of the following three sections we can not only use graphical language for topological transformations, but also introduce algebraic operations at appropriate times to simplify or complete our proofs and obtain desired results.
\begin{theorem}\label{thm2.3}
For a fixed finite-dimensional Hilbert space in $\mathbf{FHilb}$, there exists a bijection between:
$\bullet$ Orthogonal bases and commutative dagger Frobenius structures;
$\bullet$ Orthonormal bases and classical structures;\\
Specifically:
$\bullet$ Given a commutative dagger Frobenius structure $(A, \mu_{G}, \eta_{G})$, the corresponding basis consists of vectors $a \in A$ satisfying:
\begin{center}
\includegraphics[width=7cm]{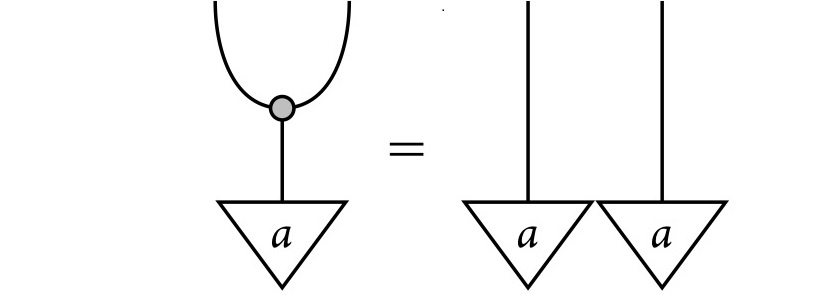}
\end{center}

$\bullet$ Given a basis $\{a_1,  \ldots, a_n\}$ of $A$. The commutative dagger Frobenius structure on $A$ can be defined as:
\begin{center}
\includegraphics[width=10cm]{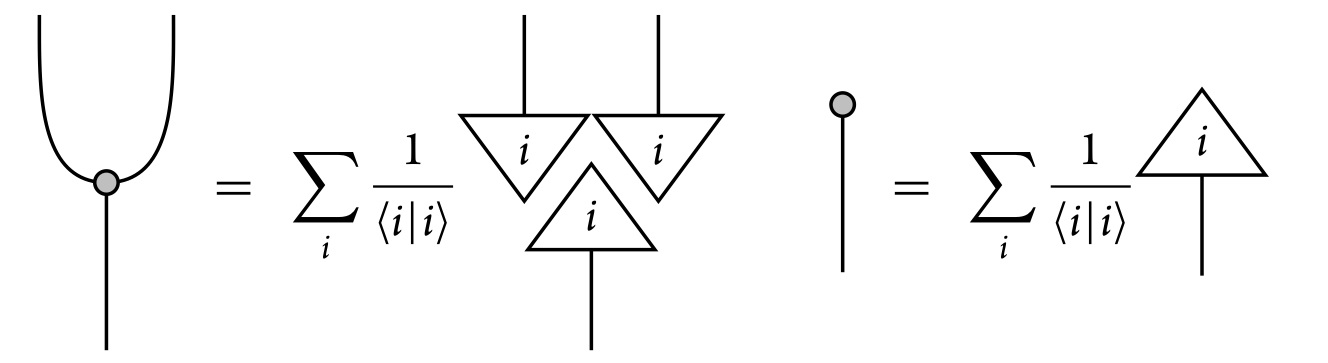}
\end{center}
\begin{center}
\includegraphics[width=10cm]{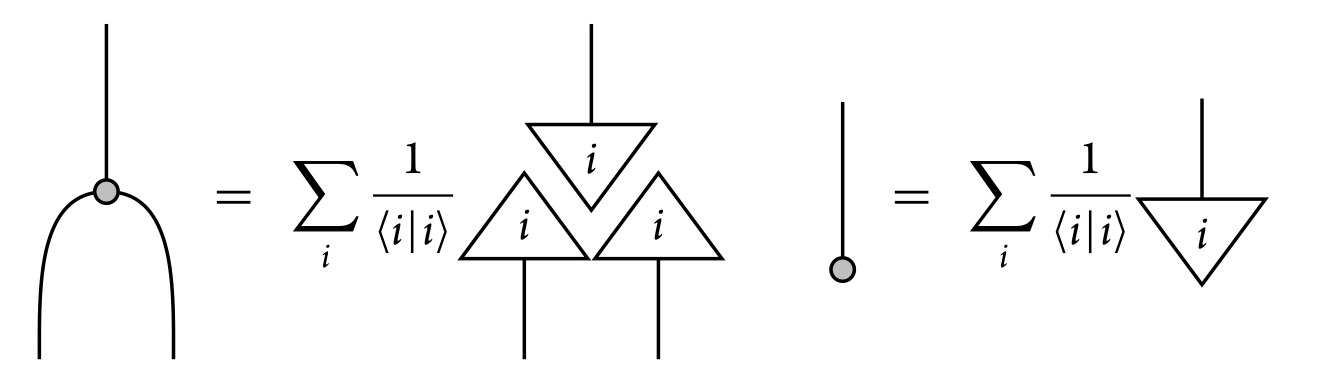}
\end{center}
\end{theorem}

Finally, in the basics of category theory, we introduce the concept of complementary Frobenius structures.
\begin{definition}[Complementary bases]
For a finite-dimensional Hilbert space $H$, two orthogonal bases $\{a_i\}$ and $\{b_j\}$ are called complementary or unbiased if there exists a constant $s \in \mathbb{C}$ such that for all $i, j$, the following holds:
$$
\left\langle a_i \mid b_j\right\rangle\left\langle b_j \mid a_i\right\rangle=s
$$
In other words, the inner products have constant absolute value.
\end{definition}
\begin{definition}[Complementary Frobenius structures]
In a braided monoidal dagger category, two symmetric dagger Frobenius structures $\mu_{B}$ and $\mu_{W}$ on the same object are called complementary if the following equality holds:
\begin{center}
\includegraphics[width=10cm]{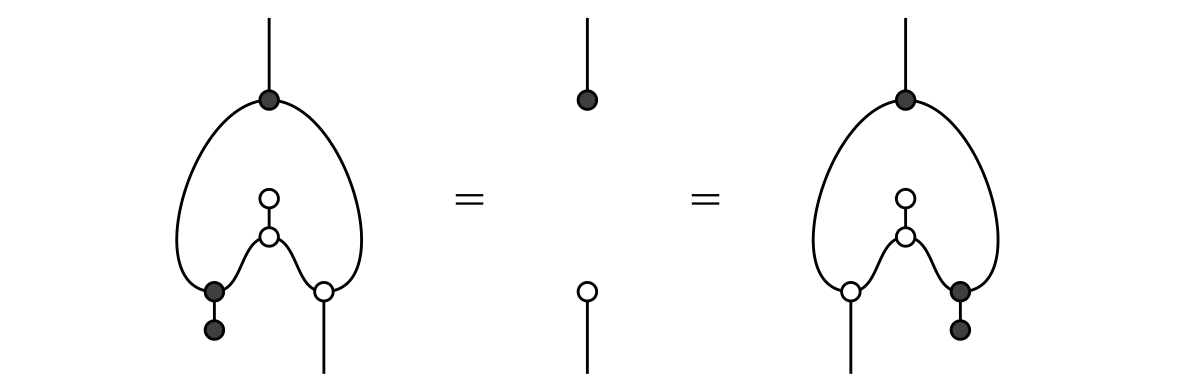}
\end{center}
The roles of black and white dots in the above definition are not obviously interchangeable. However, since the Frobenius structures are symmetric, the following rearrangement is possible:
\begin{center}
\includegraphics[width=11cm]{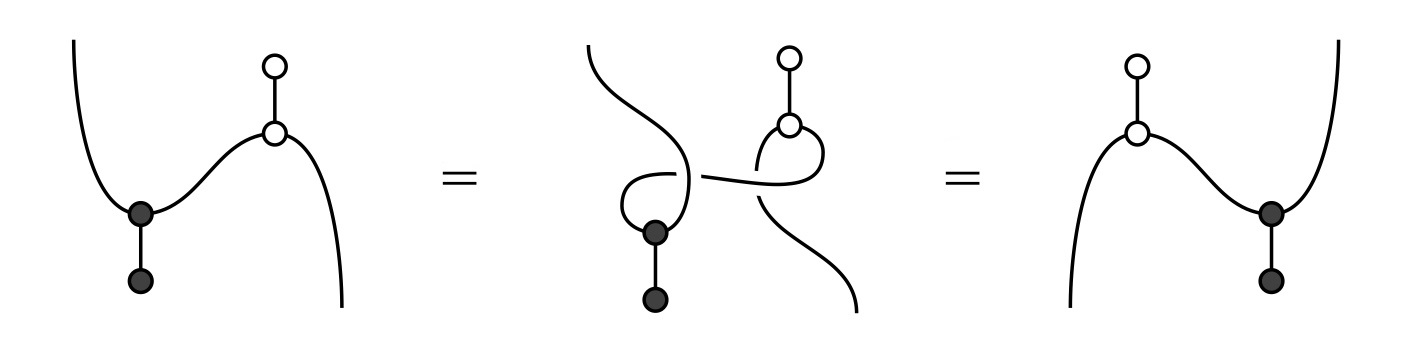}
\end{center}
Using these equations and the dagger, we can see that ``black is complementary to white" is equivalent to ``white is complementary to black".
\end{definition}
\begin{proposition}[Complementarity in $\mathbf{FHilb}$]\label{prop2.4}
In $\mathbf{FHilb}$, for two commutative dagger Frobenius structures on the same object, the following conditions are equivalent:

$\bullet$ As Frobenius structures, they are complementary;

$\bullet$ As bases, they are complementary bases with constant $s=1$.
\end{proposition}
\begin{lemma}\label{lem2.5}
In $\mathbf{FHilb}$, for a commutative dagger Frobenius structure, any copyable state a  satisfies the following equation:
\begin{center}
\includegraphics[width=10cm]{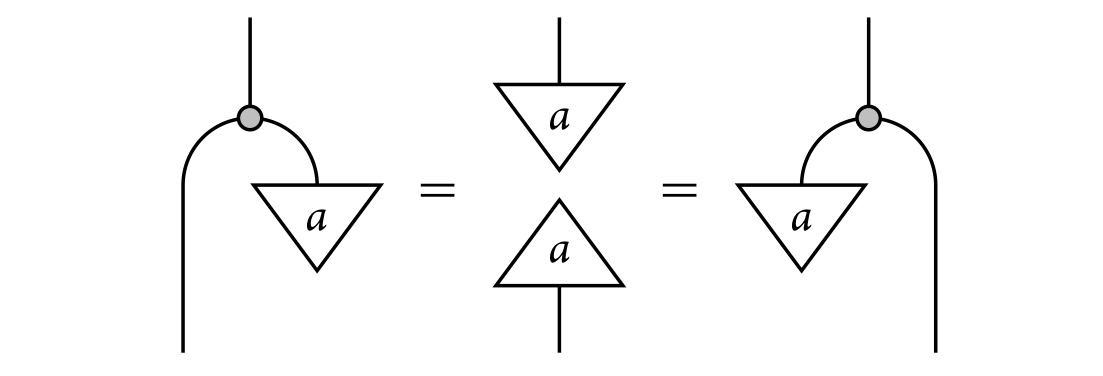}
\end{center}
\end{lemma}

\end{document}